%% file: BS-channels-arXiv-v2.tex
\documentclass[aps,prl,groupedaddress,superscriptaddress,twocolumn,notitlepage,bibnotes]{revtex4-1}

\input{newdef}

\newcommand{\icoh}{I_{\mathrm{coh}}(A\rangle B)}
\newcommand{\inv}{V}

\begin{document}

\title{Bosonic quantum communication across arbitrarily high loss channels}
%\title{Die-hard bosonic quantum communication}

\author{Ludovico Lami}
\email{ludovico.lami@gmail.com}
\affiliation{Institut f\"{u}r Theoretische Physik und IQST, Universit\"{a}t Ulm, Albert-Einstein-Allee 11, D-89069 Ulm, Germany}

\author{Martin B. Plenio}
\email{martin.plenio@uni-ulm.de}
\affiliation{Institut f\"{u}r Theoretische Physik und IQST, Universit\"{a}t Ulm, Albert-Einstein-Allee 11, D-89069 Ulm, Germany}

\author{Vittorio Giovannetti}
\email{vittorio.giovannetti@sns.it}
\affiliation{NEST, Scuola Normale Superiore and Istituto Nanoscienze-CNR, I-56127 Pisa, Italy}

\author{Alexander S. Holevo}
\email{holevo@mi-ras.ru}
\affiliation{Steklov Mathematical Institute, Gubkina 8, 119991 Moscow, Russia}

\begin{abstract}
A general attenuator $\Phi_{\lambda, \sigma}$ is a bosonic quantum channel that acts by combining the input with a fixed environment state $\sigma$ in a beam splitter of transmissivity $\lambda$. If $\sigma$ is a thermal state the resulting channel is a thermal attenuator, whose quantum capacity vanishes for $\lambda\leq 1/2$. We study the quantum capacity of these objects for generic $\sigma$, proving a number of unexpected results. Most notably, we show that for any arbitrary value of $\lambda>0$ there exists a suitable single-mode state $\sigma(\lambda)$ such that the quantum capacity of $\Phi_{\lambda,\sigma(\lambda)}$ is larger than a universal constant $c>0$. Our result holds even when we fix an energy constraint at the input of the channel, and implies that quantum communication at a constant rate is possible even in the limit of arbitrarily low transmissivity, provided that the environment state is appropriately controlled. We also find examples of states $\sigma$ such that the quantum capacity of $\Phi_{\lambda,\sigma}$ is not monotonic in $\lambda$. These findings may have implications for the study of communication lines running across integrated optical circuits, of which general attenuators provide natural models.
\end{abstract}

%For the special value $\lambda=1/2$, we show that the quantum capacity of $\Phi_{1/2,\sigma}$ vanishes whenever $\sigma$ is \tcb{a convex combination of states that are symmetric under phase space inversion up to displacements}, and in particular when $\sigma$ is Gaussian. However, we also provide an example of a state $\sigma$ that does not possess this property and that makes the quantum capacity of $\Phi_{1/2,\sigma}$ non-zero. 

\maketitle

\textbf{\em Introduction}.--- Quantum optics will likely play a major role in the future of quantum communication~\cite{KLM, Braunstein-review, CERF, weedbrook12}. Indeed, practically all quantum communication in the foreseeable future will rely on optical platforms. For this reason, the study of quantum channels acting on continuous variable (CV) systems, that is, finite ensembles of electromagnetic modes, is a core area of the rapidly developing field of quantum information~\cite{BUCCO, HOLEVO, HOLEVO-CHANNELS-2}.

In the best studied models of optical communication, one represents an optical fibre as a memoryless thermal attenuator channel. Mathematically, its action can be thought of as that of a beam splitter with a certain transmissivity $0\leq \lambda\leq 1$, where the input state is mixed with a fixed environment state $\sigma$ that is assumed to be thermal. This approximation is well justified when the signal rate is sufficiently low that memory effects are negligible, and when the optical fibre is so long that the `effective' environment state, resulting from averaging several elementary interactions that are effectively independent, due to the limited correlation length of the environment, is practically Gaussian and thermal, as follows from the quantum central limit theorem~\cite{Cushen1971, QCLT}. And indeed, an impressive amount of literature has been devoted to finding bounds on the quantum capacity of the thermal attenuator. We now have exact formulae for the zero-temperature case~\cite{holwer, Caruso2006, Wolf2007, Mark2012, Mark-energy-constrained, Noh2019}, and tight upper~\cite{PLOB, Rosati2018, Sharma2018, Noh2019} and lower~\cite{holwer, Noh2020} bounds in all other cases.

However, the thermal noise approximation is challenged when memory effects become important~\cite{memory-review}, or when the communication channel is so short that the averaging process cannot possibly take place, as may happen e.g.\ in miniaturised quantum optical circuits
%, that the promising field of integrated quantum photonics aims to exploit to implement fault-tolerant quantum computation
~\cite{OBrien2009, Politi2009, Carolan2015, Rohde2015}.
In both cases, it is conceivable that the environment state may be manipulated and engineered to facilitate communication. Namely, one could exploit memory effects to send pulses that alter it and precede the actual transmission, or one could design the integrated optical circuit that surrounds the communication line in order to control the noise that comes from other elements of the same circuit.
%A communication line connecting two sites of such a circuit may incur noise that is far away from being thermal, as it comes from active quantum elements (e.g.\ single-photon sources).
We are thus led to investigate general attenuator channels, hereafter denoted with $\Phi_{\lambda,\sigma}$, where the environment state $\sigma$ is no longer thermal. Unsurprisingly, such models have received increasing attention recently~\cite{Koenig2015, Jack2018, KK-VV, G-dilatable, Lim2019, QCLT}. As discussed above, we will be interested in optimising over the environment state so as to increase the capacity~\cite{Karumanchi2016a, Karumanchi2016b}.
%This could be accomplished by carefully designing the integrated optical circuit that surrounds our communication line, in order to control the noise that comes from other elements of the same circuit.

Other motivations for considering general attenuators stem on the one hand from the need to go beyond the Gaussian formalism to accomplish several tasks that are critical to quantum information, e.g.\ universal quantum computation~\cite{Menicucci2006, Ohlinger2010} entanglement distillation~\cite{nogo1, nogo2, nogo3}, entanglement swapping~\cite{optimal-G-ent-swapping, Namiki2014}, error correction~\cite{Niset2009}, and state transformations in general resource theories~\cite{G-resource-theories, assisted-Ryuji}. On the other hand, general attenuators are among the simplest examples of non-Gaussian channels that are nevertheless \emph{Gaussian dilatable}, meaning that they can be Stinespring dilated~\cite{Stinespring} by means of a symplectic unitary~\cite{KK-VV, G-dilatable}. This makes them amenable to a quantitative analysis in many respects. For example, it has been shown that making the environment state non-Gaussian, e.g.\ by means of a photon addition, can be advantageous when transmitting quantum or private information~\cite{KK-VV}. In spite of their increased complexity compared to Gaussian channels, the entanglement-assisted capacity of a general attenuator can nevertheless be upper bounded thanks to the conditional entropy power inequality~\cite{Koenig2015, Jack2018}. Similar bounds can be obtained for the quantum~\cite{Lim2019} and private~\cite{Jeong2020} capacity as well, by making use of the solution to the minimum output entropy conjecture~\cite{Giovadd, Giovadd-CMP, Jack-constrained} combined with known extremality properties of Gaussian states~\cite{Eisert-Wolf, Wolf2006}. 
%Finally, a model of optical fibre that employs a concatenation of general attenuators has been recently proposed~\cite{QCLT}. In that context, the associated quantum and classical capacities can be estimated thanks to a quantum Berry--Esseen inequality~\cite{QCLT}.
Finally, we have mentioned that by concatenating a large number $n$ of general attenuators with a fixed total transmissivity one typically obtains an effective channel that resembles a thermal attenuator. In this regime of large but finite $n$, the associated quantum capacity can be bounded thanks to the quantum Berry--Esseen inequality~\cite[Corollary~13]{QCLT}.

Here we investigate the quantum capacity of general attenuators $\Phi_{\lambda, \sigma}$, uncovering some unexpected phenomena. It has been observed~\cite[Lemma~16]{QCLT} that output states of general attenuators with transmissivity $\lambda=1/2$ have non-negative Wigner functions~\cite{Wigner, Hillery1984}. At first sight, this may suggest that such channels are somewhat `classical'~\cite{Hudson1974, Hudson-thm-multimode, Broecker1995}. Indeed, we show that for all convex combinations of symmetric states -- and in particular for all Gaussian states -- $\Phi_{1/2,\, \sigma}$ is anti-degradable and therefore its quantum capacity satisfies $Q\left(\Phi_{1/2,\, \sigma}\right)=0$~\cite{Devetak-Shor}. Here we call a state symmetric if it remains invariant under phase space inversion up to displacements. However, we also find an example of a state $\sigma$ that does not belong to this class and that makes $Q\left(\Phi_{1/2,\, \sigma}\right)>0$.

Next, we tackle the question of whether transmission of quantum information is possible even for very low values of the transmissivity $0<\lambda \ll 1$. Intuitively, a beam splitter of transmissivity $\lambda\leq 1/2$ should give away to the environment more than it transmits. By the no-cloning theorem, we could be led to conjecture that the quantum capacity $Q\left(\Phi_{\lambda,\sigma}\right)$ vanishes for all $\sigma$ as soon as $\lambda\leq 1/2$. Indeed, this is exactly what happens for thermal attenuators. This intuition is further supported by the analysis of general finite-dimensional depolarising channels $\Delta_{\lambda,\sigma}(\rho)$, defined by $\Delta_{\lambda,\sigma}(\rho) \coloneqq \lambda \rho + (1-\lambda)\sigma$, whose quantum capacity also vanishes for $\lambda\leq 1/2$.

However, we establish the following surprising result: \emph{for all values of $\lambda>0$ one can find suitable states $\sigma(\lambda)$ that make $Q\left( \Phi_{\lambda,\,\sigma(\lambda)}\right)\geq c$, where the constant $c>0$ is universal} (Theorem~\ref{Behemoth thm}). This implies, but is stronger than, the fact that $\Phi_{\lambda,\,\sigma(\lambda)}$ can be used to distribute entanglement~\footnote{See the Supplemental Material, which contains the references~\cite{Jordan1935, LL17, HALL-LIE, BARNETT-RADMORE, Holevo-EB, PeresPPT, 2-qubit-distillation, Horodecki-review, Masanes2006, Brunner-review, CHSH}, for complete proofs of some of the results discussed in the main text.}. As a corollary, we also see that $Q\left( \Phi_{\lambda,\,\sigma}\right)$ is in general not monotonic in $\lambda$ for fixed $\sigma$. All this marks a striking difference with the aforementioned behaviour of thermal attenuators and depolarising channels, and reveals that the phenomenology of general attenuators is richer than perhaps expected. Our proof is fully analytical, and goes by analyising the single-copy coherent information associated with a specific transmission scheme. By a tour-de-force of inequalities we show that the output state of the channel is \emph{majorised} by that of the associated complementary channel. In turn, this makes it possible to lower bound the coherent information by applying a beautiful inequality recently proved by Ho and Verd\'{u}~\cite{Ho2010}.

\textbf{\em Notation}.--- 
The Hilbert space corresponding to an $m$-mode CV comprises all square-integrable functions $\R^{m}\to \C$, and is denoted by $\HH_m \coloneqq L^2(\R^m)$. Quantum states are represented by \emph{density operators} on $\HH_m$, i.e.\ positive semi-definite trace class operators with unit trace. We will denote with $a_j,a_j^\dag$, respectively, the \emph{annihilation} and \emph{creation} operators corresponding to the $j$-th mode, and with $\ket{0}$ the \emph{vacuum state}. The \emph{canonical commutation relations} read $[a_j, a_k^\dag] = \delta_{jk} I$, $[a_j,a_k] = 0$. The unitary \emph{displacement operators} on $\HH_m$ are constructed as $D(\alpha)\coloneqq e^{\sumno_j (\alpha_j a_j^\dag - \alpha_j^* a_j)}$, where $\alpha\in \C^m$; they satisfy $D(\alpha) D(\beta) = e^{\frac12 (\alpha^\intercal \beta^* - \alpha^\dag \beta)} D(\alpha+\beta)$ for all $\alpha,\beta\in \C^m$.

For every trace class operator $T$ on $\HH_m$, its \emph{characteristic function} $\chi_T:\C^{m}\to \C$ is defined by~\cite{HOLEVO, Werner84}
\bb
\chi_T(\alpha)\coloneqq \Tr[T D(\alpha)]\, .
\label{chi}
\ee
The \emph{Wigner function} $W_T$ of $T$ is the Fourier transform of $\chi_T$~\cite{Wigner, Hillery1984, Werner84, HOLEVO}.
%\bb W_T(\alpha) \coloneqq \int \frac{d^{2m}\beta}{\pi^{2m}}\ \chi_\rho(\beta)\, e^{\alpha\beta^* - \alpha^* \beta}\, . \label{W} \ee
Note that $W_\rho$ is typically not pointwise positive for a generic quantum state $\rho$~\cite{Hudson1974, Hudson-thm-multimode, Broecker1995}. 

%In fact, a famous result by Hudson states that the only pure states with a positive Wigner function are Gaussian states~\cite{Hudson1974, Hudson-thm-multimode} (see also~\cite{Broecker1995}). 

%As an example, the Wigner function of the Fock state $\ket{n}$ is given by~\cite[Eq.~(4.5.31)]{BARNETT-RADMORE}
%\bb W_{\ket{n}\!\bra{n}}(\alpha) = \frac{2(-1)^n}{\pi}\, e^{-2|\alpha|^2}\, L_n(4|\alpha|^2)\, , \label{Wigner Fock} \ee
%where $L_n$ is the $n$-the Laguerre polynomial.

%Unitaries that can be represented as finite products of factors of the form $V=e^{-iK}$, where $K$ is an arbitrary (inhomogeneous) quadratic expression in the canonical operators, are called \emph{symplectic}. A particularly simple symplectic unitary is the operator $U_\lambda$ associated with a $(m\!+\!m)$-mode \emph{beam splitter} of transmissivity $\lambda\in [0,1]$, given by
A \emph{beam splitter} of transmissivity $0\leq \lambda\leq 1$ acting on two systems of $m$ modes each is represented by the unitary operator
\bb
U_\lambda \coloneqq e^{\arccos\sqrt{\lambda}\, \sumno_j ( a_j^\dag b_j - a_j b_j^\dag)}\, ,
\label{BS}
\ee
where $a_j,b_j$ are the annihilation operators on the $j$-th modes of the first and second system, respectively. Our main object of study is the \emph{general attenuator} channel $\Phi_{\lambda, \sigma}$, which acts on an $m$-mode system $B$ as
\bb
\Phi_{\lambda, \sigma}^B(\rho_B) \coloneqq \Tr_E\left[ U_\lambda^{BE} (\rho_B\otimes \sigma_E) \left(U_\lambda^{BE}\right)^\dag \right] .
\label{Phi}
\ee
Dropping the system labels for simplicity, this can be cast in the language of characteristic functions as
\bb
\chi_{\Phi_{\lambda, \sigma}(\rho)}(\alpha) = \chi_\rho\left(\sqrt\lambda\, \alpha\right) \chi_\sigma\left(\sqrt{1-\lambda}\, \alpha\right) .
\label{Phi chi}
\ee
A pictorial representation of the action of a general attenuator is provided in Figure~\ref{general attenuator}.

The \emph{thermal attenuators} $\mathcal{E}_{\lambda, \nu} \coloneqq \Phi_{\lambda,\, \tau_\nu}$ as well as the \emph{pure loss channels} $\mathcal{E}_{\lambda}\coloneqq \mathcal{E}_{\lambda, 0}=\Phi_{\lambda,\, \ket{0}\!\bra{0}}$ are standard examples of single-mode attenuators, obtained by taking the environment to be in a thermal state $\tau_\nu\coloneqq \frac{1}{\nu+1}\sum_{n=0}^\infty \left( \frac{\nu}{\nu+1}\right)^n \ketbra{n}$, where %$\ket{n}\coloneqq (n!)^{-1/2} (a^\dag)^n \ket{0}$
$\ket{n}$ is the $n$-th \emph{Fock state}.

\begin{figure}
\centering
\includegraphics[scale=0.8]{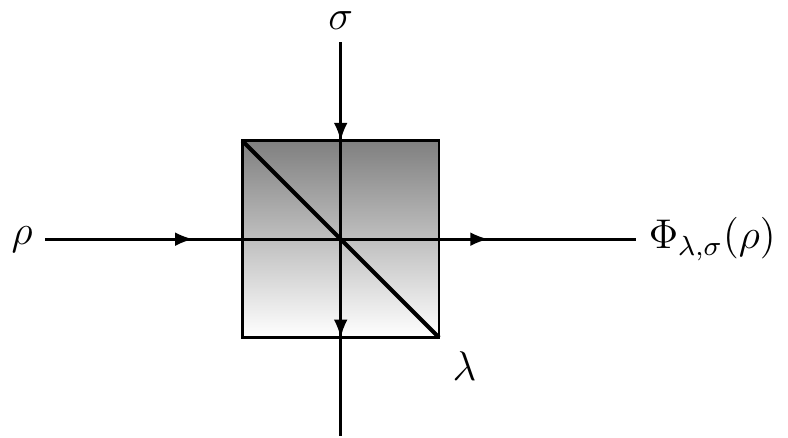}
\caption{A general attenuator acts by mixing the input state $\rho$ in a beam splitter of transmissivity $\lambda$ with an environment in a fixed state $\sigma$.}
\label{general attenuator}
\end{figure}

Quantum channels are useful because they can transmit quantum information. The maximum rate at which independent copies of a channel $\Phi$ acting on a system $B$ can simulate instances of the noiseless qubit channel $I_2$ is called the \emph{quantum capacity} of $\Phi$, and denoted with $Q(\Phi)$. For CV systems, physical transmission of quantum data must be subjected to an energy constraint. We shall assume that the relevant Hamiltonian is the total photon number: for an $m$-mode system, $H_m \coloneqq \sum_{j=1}^m a_j^\dag a_j$. The \emph{energy-constrained quantum capacity} can be obtained thanks to the following modified version~\cite[Theorem~5]{Mark-energy-constrained} of the Lloyd--Shor--Devetak theorem~\cite{Lloyd-S-D, L-Shor-D, L-S-Devetak, HOLEVO-CHANNELS}:
%\begin{align}
%Q(\Phi) =&\ \sup_{k} \frac1k\, Q^{(1)}\!\left( \Phi^{\otimes k}\right) , \label{Q 1} \\
%Q_1\! (\Phi) \coloneqq&\ \sup_{\ket{\Psi}_{\!AB}} \icoh_{(I_A\otimes \Phi_B)(\Psi_{AB})}\, , \label{Q 2}
%\end{align}
\begin{align}
Q\left(\Phi, N\right) =&\ \sup_{k} \frac1k\, Q_1\!\left(\Phi^{\otimes k},\, kN\right) , \label{energy-constrained Q 1}\\
Q_1\! \left(\Phi, N\right) \coloneqq& \sup_{\Tr[\Psi_B H_B]\leq N} \!\! \icoh_{(I_A\otimes \Phi_B)(\Psi_{AB})}\, . \label{energy-constrained Q 2}
\end{align}
where $\Psi_{AB}\coloneqq \ketbra{\Psi}_{AB}$ is pure, and $\icoh_\rho\coloneqq \Tr\left[ \rho_{AB} \left( \log_2 \rho_{AB} - \log_2 \rho_B \right) \right]$ is the coherent information. The unconstrained quantum capacity is obtained as $Q\left(\Phi\right)\coloneqq \lim_{N\to\infty} Q\left(\Phi,N\right)$. In general, the expression in~\eqref{energy-constrained Q 1} is intractable. However, for the pure loss channel the regularisation is not needed, and the quantum capacity can be expressed in closed form as~\cite{holwer, Wolf2007, Mark2012, Mark-energy-constrained, Noh2019}
\bb
Q\left( \mathcal{E}_\lambda, N\right) = \max\left\{ g(\lambda N) - g((1-\lambda) N),\, 0\right\} ,
\label{pure loss Q}
\ee
where $g(x)\coloneqq (x+1)\log_2 (x+1) - x\log_2 x$ is the bosonic entropy. No such formula is known for the thermal attenuators, although sharp bounds are available~\cite{holwer, PLOB, Rosati2018, Sharma2018, Noh2019, Noh2020}.

\textbf{\em Results}.---
Before expounding our findings, let us forge our intuition by looking at other channels that present some analogies with general attenuators. An obvious starting point is the thermal attenuator $\mathcal{E}_{\lambda,\nu} = \Phi_{\lambda, \tau_\nu}$. When $\lambda\leq 1/2$, $\mathcal{E}_{\lambda,\nu}$ is anti-degradable, meaning that tracing out $B$ instead of $E$ in~\eqref{Phi} results in a channel that can simulate $\mathcal{E}_{\lambda, \nu}$ via post-processing~\cite{Devetak-Shor, Caruso2006, extendibility}. This implies that $Q(\mathcal{E}_{\lambda,\nu})=0$ for $\lambda\leq 1/2$~\cite[p.~3]{Caruso2006}. %In fact, all $k$-extendibility regions have been precisely described~\cite{extendibility}.
On a different note, we can also consider a generalised depolarising channel in finite dimension $d$, acting as $\rho\mapsto \Delta_{\lambda,\sigma}(\rho) = \lambda \rho + (1-\lambda)\sigma$. As it turns out, its quantum capacity is again zero for $\lambda\leq 1/2$. In fact, $\Delta_{\lambda,\sigma}$ can be obtained from an erasure channel~\footnote{An erasure channel acts as $\rho\mapsto \mathcal{N}_\lambda(\rho)\coloneqq \lambda \rho + (1-\lambda) \ketbra{e}$, where $\ket{e}$ is an error flag that is orthogonal to every input state. Constructing the post-processing channel $\rho \mapsto \mathcal{M}_\sigma (\rho) \coloneqq \left(\id - \ketbra{e}\right) \rho \left(\id-\ketbra{e}\right) + \braket{e|\rho|e} \sigma$, we see that $\Delta_{\lambda,\sigma} = \mathcal{M}_\sigma \circ \mathcal{N}_\lambda$.} via post-processing. Since the quantum capacity of this latter object is known~\cite{erasure}, by data processing we obtain that $Q\left( \Delta_{\lambda,\sigma}\right)\leq \max\left\{(1-2\lambda)\log_2 d,\, 0\right\}$ for all $\sigma$. In particular, $Q\left( \Delta_{\lambda,\sigma}\right)=0$ for $\lambda\leq 1/2$.

Our results show that the phenomenology of general attenuators is way richer than these considerations may have suggested. We start by looking at the role of the special point $\lambda=1/2$.

%This fact can be intuitively understood by imagining that a low-transmissivity thermal attenuator gives away to the environment more quantum information than it transmits, and by applying  

\begin{thm} \label{symmetric sigma thm}
Let $\sigma$ be an $m$-mode state of the form $\sigma = \int d\mu(\alpha)\,  D(\alpha) \sigma_0(\alpha) D(\alpha)^\dag$, where $\alpha\in \C^m$, $\mu$ is a probability measure on $\C^m$, and the states $\sigma_0(\alpha) = \inv \sigma_0(\alpha) \inv^\dag$ are symmetric under the phase space inversion operation $\inv \coloneqq (-1)^{H_m}$, with $H_m$ being the total photon number. Then the channel $\Phi_{1/2,\, \sigma}$ is anti-degradable~\cite{Devetak-Shor}, and in particular $Q\left( \Phi_{1/2,\, \sigma}\right) = 0$.
\end{thm}

\begin{proof}
Under our assumptions it holds that $\Phi_{1/2,\,\sigma} = \int d\mu(\alpha)\, \Phi_{1/2,\, D(\alpha) \sigma_0(\alpha) D(\alpha)^\dag}$. Now, since the set of anti-degradable channels is convex~\cite[Appendix~A.2]{Cubitt2008}, we can directly assume that $\mu$ is a Dirac measure, i.e.\ $\sigma=D(\alpha) \sigma_0 D(\alpha)^\dag$ with $\sigma_0$ symmetric under phase space inversion.
Acting on $\rho \otimes \sigma$ with the beam splitter unitary $U_\lambda$ yields a global state with characteristic function
\bbb
\chi_\rho\left( \sqrt\lambda \alpha - \sqrt{1-\lambda} \beta \right) \chi_\sigma \left( \sqrt{1-\lambda}\alpha + \sqrt\lambda \beta\right) .
\eee
While the reduced state on the first system is given by~\eqref{Phi chi}, that on the second system has characteristic function $\chi_{\rho}\big( -\sqrt{1-\lambda} \beta \big)\chi_\sigma \big( \sqrt\lambda \beta\big)$, which coincides with that of $\inv \Phi_{1-\lambda,\, \inv \sigma\inv^\dag}\left( \rho \right) \inv^\dag$. Therefore, the weak complementary channel associated to $\Phi_{\lambda, \sigma}$ via the representation~\eqref{Phi} can be expressed as
\bbb
\Phi_{\lambda, \sigma}^{\mathrm{wc}} = \pazocal{V}\circ \Phi_{1-\lambda,\, \pazocal{V} (\sigma)}\, ,
\eee
where $\pazocal{V}(\cdot)\coloneqq \inv (\cdot) \inv^\dag$. 

Using the identity $V D(\alpha) V^\dag = D(-\alpha)$, we see that when $\sigma = D(\alpha)\sigma_0 D(\alpha)^\dag$ we also have that $\pazocal{V}(\sigma) = \pazocal{D}_{-2\alpha}(\sigma)$, where $\pazocal{D}_{z}(\cdot)\coloneqq D(z)(\cdot) D(z)^\dag$. Noting that $\Phi_{1-\lambda,\, \pazocal{D}_z(\sigma)} = \pazocal{D}_{\sqrt{\lambda}z}\circ \Phi_{1-\lambda,\,\sigma}$, we finally obtain that
\bbb
\Phi_{\lambda, \sigma}^{\mathrm{wc}} = \pazocal{V}\circ \pazocal{D}_{-2\sqrt\lambda \alpha}\circ \Phi_{1-\lambda,\, \sigma}\, .
\eee
Thus, if $\lambda=1/2$ the channel is equivalent to its weak complementary up to a unitary post-processing.
\end{proof}

The class of states $\sigma$ to which Theorem~\ref{symmetric sigma thm} applies is invariant under symplectic unitaries and displacement operators, and it includes many states that are relevant for applications, for instance all convex combinations of Gaussian states (e.g.\ \emph{classical states}~\cite{Bach1986, Yadin2018}) and all Fock-diagonal states.
Remarkably, the above result no longer holds if we weaken the assumption on $\sigma$. To see this, for $0\leq \eta\leq 1$ consider the family of single-mode states $\xi(\eta)=\ketbra{\xi(\eta)}$, with $\ket{\xi(\eta)} \coloneqq \sqrt{\eta}\ket{0} - \sqrt{1-\eta} \ket{1}$. A lower bound on the energy-constrained quantum capacity of the channels $\Phi_{1/2,\, \xi(\eta)}$ can be obtained by setting
$\ket{\Psi(\eta)}_{AB} \coloneqq \sqrt{\eta (1\!-\!\eta)}\ket{00} + (1\!-\!\eta)\ket{01}+\sqrt{\eta}\ket{10}$
and by considering that~\cite{Note1}
\bb
Q\left( \Phi_{1/2,\,\xi(\eta)},\, (1\!-\!\eta)^2\right) \geq I_{\mathrm{coh}}(A\rangle B)_{\zeta_{AB}(\eta)}\, ,
\label{Icoh}
\ee
where $\zeta_{AB}(1/2, \eta) \coloneqq \big(I^A\otimes \Phi_{1/2,\, \xi(\eta)}^B \big)(\Psi_{AB}(\eta))$, and $\Psi(\eta)\coloneqq \ketbra{\Psi(\eta)}$. The function on the r.h.s.\ of~\eqref{Icoh} is strictly positive for all $0<\eta<1$~\cite{Note1}.

The above example shows that quantum communication can be possible on a general attenuator even for transmissivity $\lambda=1/2$. At this point, we may wonder whether at least for a fixed energy constraint at the input there exists a threshold value for $\lambda$ below which quantum communication becomes impossible. Our main result states that this is not the case; on the contrary, the quantum capacity can be bounded away from $0$ even when $\lambda$ approaches $0$, if the environment state $\sigma$ is chosen appropriately. Note that the bounds by Lim et al.~\cite{Lim2019} cannot possibly be used to draw such a conclusion~\cite{Note1}.

\begin{thm} \label{Behemoth thm}
For all $0<\lambda\leq 1$ there exists a single-mode (pure) state $\sigma(\lambda)$ such that
\bb
Q\big(\Phi_{\lambda,\,\sigma(\lambda)}\big) \geq Q\big(\Phi_{\lambda,\,\sigma(\lambda)},\, 1/2 \big) \geq c
\label{Behemoth}
\ee
for some universal constant $c>0$. %Both $\sigma(\lambda)$ and $c$ are explicitly given in the proof.
Depending on $\lambda$, we can take $\sigma(\lambda)$ to be either the vacuum $\ket{0}$, or a superposition $\alpha \ket{0} + \beta \ket{1}$, or a Fock state $\ket{n}$ with $n\geq 2$.
\end{thm}

\begin{proof}[Sketch of the proof]
When $1/2 < \lambda\leq 1$, it suffices to set $\sigma(\lambda)=\ketbra{0}$ and leverage~\eqref{pure loss Q}. Around $\lambda=1/2$, positive quantum capacity follows by perturbing the lower bound in~\eqref{Icoh} thanks to the Alicki--Fannes--Winter inequality~\cite{Alicki-Fannes, tightuniform}. It remains to establish the result for $0<\lambda\leq 1/2-\epsilon$, where $\epsilon>0$ is fixed. We start by making an ansatz for a state $\ket{\Psi}_{AB}$ to be plugged into~\eqref{energy-constrained Q 2}. Let us set $\ket{\Psi}_{AB} \coloneqq \frac{1}{\sqrt2}\big( \ket{01} + \ket{10}\big)$ and $\sigma(n)\coloneqq \ketbra{n}$. The output state $\omega_{AB}(n,\lambda) \coloneqq \big( I^A\otimes \Phi_{\lambda,\, \sigma(n)}^B\big) (\Psi_{AB})$ can be computed e.g.\ thanks to the formulae derived by Sabapathy and Winter~\cite[Section~III.B]{KK-VV}. One obtains that
\begin{align*}
Q\big( \Phi_{\lambda,\, \sigma(n)}, 1/2\big) \geq \mathcal{I}(n,\lambda) \coloneqq&\ \icoh_{\omega_{AB}(n,\lambda)} \\
=&\ H\left(p(n,\lambda)\right) - H\left(q(n,\lambda)\right) ,
\end{align*}
where the two probability distributions $p(n,\lambda)$ and $q(n,\lambda)$ over the alphabet $\{0,\ldots, n+1\}$ are defined by
\begin{align*}
p_\ell(n,\lambda) &\coloneqq \frac{1}{2(n\!+\!1)(1\!-\!\lambda)} \binom{n\!+\!1}{\ell} \left(1\!-\!\lambda\right)^{\ell}\lambda^{n-\ell} \\
&\quad\ \times \left((1\!-\!\lambda)(n\!-\!\ell\!+\! 1) + \left( (n\!+\!1)(1\!-\!\lambda) - \ell \right)^2\right) , \\
q_\ell(n,\lambda) &\coloneqq  \frac{1}{2(n\!+\!1)(1\!-\!\lambda)} \binom{n\!+\!1}{\ell} \left(1\!-\!\lambda\right)^{\ell}\lambda^{n-\ell} \\
&\quad\ \times \left( \lambda\ell + \left( (n\!+\!1)(1\!-\!\lambda) - \ell \right)^2\right) .
\end{align*}

\begin{figure}[ht]
\centering
\includegraphics[scale=0.85]{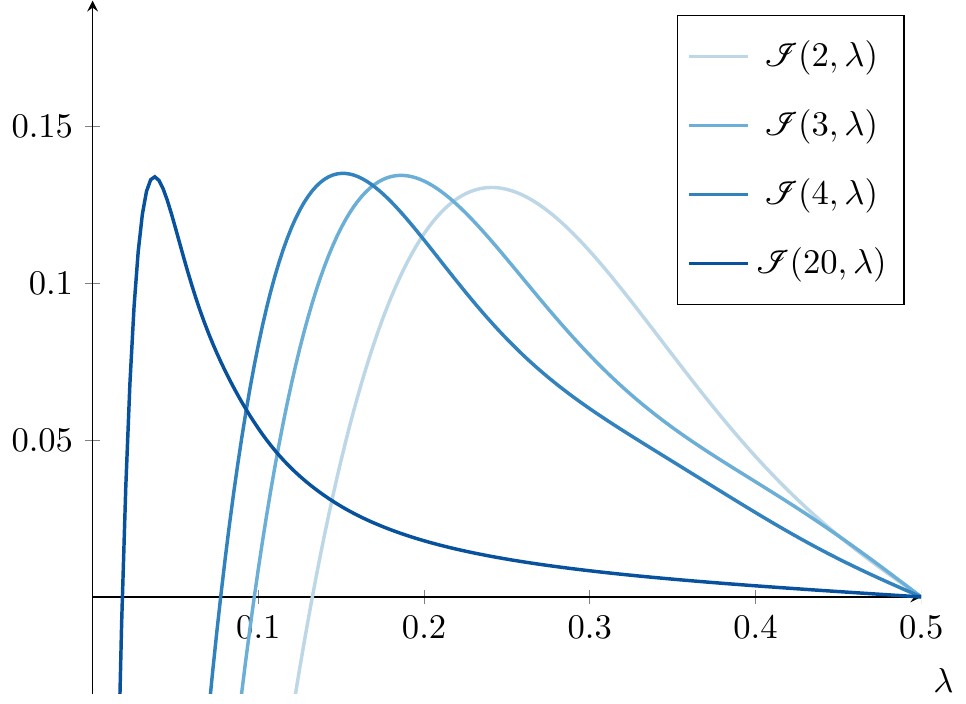}
\caption{The functions $\mathcal{I}(n,\lambda)$ plotted with respect to the variable $\lambda$ for several values of $n$.}
\label{Icoh figure}
\end{figure}

In Figure~\ref{Icoh figure} we plotted $\mathcal{I}(n,\lambda)$ as a function of $\lambda$ for increasing values of $n$. The lower endpoint of the range for which $\mathcal{I}(n,\lambda)\geq c$ for some fixed $c>0$ seems to move closer and closer to $0$ as $n$ grows. However, an analytical proof of this fact is technically challenging.
The crux of our argument is to show that $p(n,\lambda)$ and $q(n,\lambda)$ are in a majorisation relation, that is, $p(n,\lambda)\prec q(n,\lambda)$ for all $n\geq 2$ and all $\frac{1}{n+1}\leq \lambda\leq \frac1n$. Given two probability distributions $r$ and $s$ over the same alphabet $\{0,\ldots, N\}$, we say that $r$ is \emph{majorised} by $s$, and we write $r\prec s$, if $\sum_{\ell=0}^k r^\uparrow_\ell \geq \sum_{\ell=0}^k s^\uparrow_\ell$ holds for all $k=0,\ldots, N$, where $r^\uparrow$ and $s^\uparrow$ are obtained by sorting $r$ and $s$ in ascending order~\cite{MARSHALL-OLKIN}. This definition captures the intuitive notion of $r$ being `more disordered' than $s$. An immediate consequence is that the entropy of $r$ is never smaller than that of $s$. But more is true: a beautiful inequality recently established by Ho and Verd\'u~\cite[Theorem~3]{Ho2010} allows us to lower bound the entropy difference as
\bb
H(s) - H(r) \geq D\left(s^\uparrow \big\|\, r^\uparrow\right)\, ,
\label{Ho-Verdu}
\ee
where $D(u\|v) \coloneqq \sum_\ell u_\ell \log_2 \frac{u_\ell}{v_\ell}$ is the Kullback--Leibler divergence. This latter quantity can be in turn lower bounded as $D(u\| v)\geq \frac{1}{2\ln 2}\left\|u-v\right\|_1^2$ in term of the \emph{total variation distance} $\|u-v\|_1\coloneqq \sum_\ell |u_\ell - v_\ell |$ thanks to Pinsker's inequality~\cite{CSISZAR-KOERNER}. We find that
\begin{align*}
\mathcal{I}\left(n,\lambda\right) &= H(p(n,\lambda)) - H(q(n,\lambda)) \\
&\geq D\big( q^\uparrow (n,\lambda) \big\| p^\uparrow (n,\lambda) \big) \\
&\geq \frac{1}{2 \ln 2} \left\|q^\uparrow (n,\lambda) - p^\uparrow (n,\lambda) \right\|_1^2 \\
&\geq \frac{2}{\ln 2} \left|q^\uparrow_{n+1} (n,\lambda) - p^\uparrow_{n+1} (n,\lambda) \right|^2 \\
&= \frac{2}{\ln 2} \left|p_{n-1} (n,\lambda) - q_{n+1} (n,\lambda) \right|^2 \, ,
\end{align*}
where in the last line we used the fact, proven in the SM~\cite{Note1}, that $p_{n-1} (n,\lambda) = \max_\ell p_{\ell} (n,\lambda)$ and $q_{n+1} (n,\lambda) = \max_\ell q_{\ell} (n,\lambda)$ for all $n\geq 2$ and $\frac{1}{n+1}\leq \lambda\leq \frac{1}{n}$. It remains to lower bound $k(n,\lambda) \coloneqq \left|p_{n-1} (n,\lambda) - q_{n+1} (n,\lambda) \right|$, which is done by inspection. We find that (a)~$k(2,\lambda) \geq \epsilon/4$ for all $1/3\leq \lambda\leq 1/2-\epsilon$; and (b)~$k(n,\lambda)\geq c$ for some universal constant $c>0$ for all $n\geq 3$ and $\frac{1}{n+1}\leq \lambda\leq \frac{1}{n}$, concluding the proof. %This concludes the proof.
\end{proof}

Note that $Q\left(\Phi_{1/2,\, \ket{n}\!\bra{n}}\right)\equiv 0$ for all $n$ by Theorem~\ref{symmetric sigma thm}, while we have just shown that $Q\left(\Phi_{\lambda,\, \ket{n}\!\bra{n}}\right)> 0$ when $\frac{1}{n+1}\leq \lambda\leq \frac1n$. This illustrates the rather surprising fact that $Q\left(\Phi_{\lambda, \sigma}\right)$ can happen not to be monotonic in $\lambda$ for a fixed $\sigma$. In the SM~\cite{Note1} we prove that monotonicity still holds under certain circumstances, e.g.\ when $\sigma=\sigma_{\G}$ is Gaussian. Combining this with Theorem~\ref{symmetric sigma thm} also shows that $Q\left( \Phi_{\lambda,\sigma_{\G}} \right) \equiv 0$ for all $\lambda\leq 1/2$ and all Gaussian $\sigma_\G$.

From the proof we see that while the energy of the input of the channel in Theorem~\ref{Behemoth thm} is fixed, that of the environment state diverges as $\lambda$ approaches $0$. Intuitively, this may be due to the need for the receiver to distinguish the faint low-energy signals, which requires environmental states with highly oscillatory phase space structures and thus high energy. Whether this reasoning can be made rigorous is left as an open problem.

We now look at the optimal value of the constant $c$ in~\eqref{Behemoth}. Our argument yields $c \geq 5.133 \times 10^{-6}$, while numerical investigations suggest that $c\gtrsim 0.066$. If only sufficiently small values of $\lambda$ are taken into account, we can prove that $c \geq 0.0244$. To put this into perspective, elementary considerations show that $c\leq 1.377$~\cite{Note1}.

\textbf{\em Conclusions}.--- We have studied the transmission of quantum information on general attenuator channels, which are among the simplest examples of non-Gaussian channels and may be relevant for applications. We have shown that their quantum capacity vanishes for transmissivity $1/2$ and for a wide class of environment states. At the same time, we have uncovered an unexpected phenomenon: namely, for any non-zero value of the transmissivity there exists an environment state that makes the quantum capacity of the corresponding general attenuator larger than a universal constant. This also implies that said quantum capacity is not necessarily monotonically increasing in the transmissivity for a fixed environment state.

\textbf{\em Acknowledgments.}--- LL and MBP are supported by the ERC Synergy Grant BIOQ (grant no.\ 319130). VG acknowledges support by MIUR via PRIN 2017 (Progetto di Ricerca di Interesse Nazionale): project QUSHIP (2017SRNBRK).

\bibliographystyle{apsrev4-1}
\bibliography{../../biblio}

%%%%%%%%%%%%%%%%%%%%%%%%%%%%%%%%%%%%%%%%%%%%%%%%%%%%%%%%%%%%%%%%%%%%%%%%%%%%%%%

\clearpage

\onecolumngrid
\begin{center}
\vspace*{\baselineskip}
{\textbf{\large Supplemental material:\\ Bosonic quantum communication across arbitrarily high loss channels}}\\
\end{center}
%\twocolumngrid

\renewcommand{\theequation}{S\arabic{equation}}
\renewcommand{\thethm}{S\arabic{thm}}
\setcounter{equation}{0}
\setcounter{thm}{0}
\setcounter{figure}{0}
\setcounter{table}{0}
\setcounter{section}{0}
\setcounter{page}{1}
\makeatletter

\setcounter{secnumdepth}{2}

\section{Generalities}

\subsection{Quantum entropy}

The (von Neumann) \emph{entropy} of a quantum state $\rho$ is defined as
\bb
S(\rho) \coloneqq - \tr\left[\rho\log_2 \rho\right] ,
\label{entropy}
\ee
which is well defined although possibly infinite. Indeed, one way to understand it is via the infinite sum $S(\rho) = \sum_i (- p_i \log_2 p_i)$, where $\rho=\sum_i p_i \ketbra{e_i}$ is the spectral decomposition of $\rho$. Since all terms in the above sum are non-negative, the sum itself is well defined but possibly infinite.

Consider an $m$-mode system with Hilbert space $\HH_m = L^2\left( \R^m\right) \simeq \HH_1^{\otimes m}$. The \emph{total photon number} is a densely defined operator on $\HH_m$ that takes the form 
\bb
H_m \coloneqq \sum_{j=1}^m a_j^\dag a_j
\label{H}
\ee
when written in terms of the creation and annihilation operators. It is well known to have a discrete spectrum of the form $\left\{ \sum_{j=1}^m n_j:\, n_j\in \N\right\}$, with the eigenvector corresponding to $\sum_{j=1}^m n_j$ being given by the tensor product of Fock states $\ket{n_1}\ldots \ket{n_m}$. The single-mode thermal state with mean photon number $\nu\geq 0$ is given by
\bb
\tau_\nu \coloneqq \frac{1}{\nu+1}\sum_{n=0}^\infty \left( \frac{\nu}{\nu+1}\right)^n \ketbra{n}\, .
\label{thermal}
\ee
The thermal state over $m$ modes with \emph{total} mean photon number $\nu$ can be easily obtained as the $m$-fold tensor product $\tau_{\nu/m}^{\otimes m}$. Thermal states are important because they are the maximisers of the entropy among all states with a fixed mean photon number. That is,
\bb
\max\left\{ S(\rho):\, \Tr\left[ \rho\, H_m\right] \leq \nu\right\}  = S\left(\tau_{\nu/m}^{\otimes m}\right) = m\, g\left(\frac{\nu}{m}\right)
\label{tau maximises entropy}
\ee
holds for all $\nu\geq 0$, where
\bb
g(x)\coloneqq (x+1)\log_2 (x+1) - x\log_2 x\, ,
\label{g}
\ee
sometimes called the \emph{bosonic entropy}, expresses the entropy of a single-mode thermal state in terms of its mean photon number. The function $g$ has many notable properties: (a) it is monotonically increasing; (b) it is subadditive, meaning that
\bb
g(x+y) \leq g(x) + g(y) \qquad \forall\ x,y\geq 0\, ;
\label{subadditivity g}
\ee
(c) it is concave; and (d) it has the asymptotic behaviour
\bb
g(x) = \log_2 (ex) + o(1)\qquad (x\to\infty)\, .
\label{asymptotics g}
\ee

\subsection{Beam splitters}

A beam splitter is perhaps the simplest example of a passive unitary acting on an $(m+m)$-mode bipartite CV system. As reported in the main text~\eqref{BS}, it is defined by $U_\lambda \coloneqq e^{\arccos\sqrt{\lambda}\, \sumno_j ( a_j^\dag b_j - a_j b_j^\dag)}$, where $a_j,b_j$ are the annihilation operators on the $j$-th modes belonging to the first and second system, respectively. This exponential can be decomposed thanks to a well-known trick. Consider the annihilation operators $a_1,\ldots, a_m$ of $m$ independent modes. The \emph{Jordan map}~\cite{Jordan1935}
\bb
J:\ X\longmapsto \sum_{j,k=1}^m X_{jk} a_j^\dag a_k\, ,
\label{Jordan map}
\ee
is a Lie algebra isomorphism between the set of $m\times m$ matrices and that of the operators on the Hilbert space $\HH_m$ of $m$ modes that are bilinear in the $a_j^\dag$ and $a_k$. Let us note in passing that the Jordan map~\eqref{Jordan map} can be extended so as to include all operators that can be expressed as polynomials of degree up to $2$ in the creation and annihilation operators~\cite[Appendix~A]{LL17}. Now, since the Baker--Campbell--Hausdorff formula shows that $\ln (e^X e^Y)$ just depends on (nested) commutators between elements of the Lie algebra generated by $X$ and $Y$, one has that~\cite[Corollary~3.4]{HALL-LIE}
\bb
J\left( \ln (e^X e^Y) \right) = \ln \left(e^{J(X)} e^{J(Y)}\right) .  
\ee
In other words, if $e^X e^Y = e^Z$ as matrices, then $e^{J(X)} e^{J(Y)} = e^{J(Z)}$ as operators. This is extremely useful in practical computations. In our case, the exponential that defines a two-mode beam splitter involves only two independent operators $a$ and $b$. Therefore, the matrix Lie algebra that corresponds to it via the Jordan map is composed of $2\times 2$ matrices. We obtain the explicit correspondence
\begin{align}
a^\dag b &\longleftrightarrow \left( \begin{smallmatrix} 0 & 1 \\ 0 & 0 \end{smallmatrix}\right) ,\\
ab^\dag &\longleftrightarrow \left( \begin{smallmatrix} 0 & 0 \\ 1 & 0 \end{smallmatrix}\right) ,\\
\frac12(a^\dag a - b^\dag b) &\longleftrightarrow \left( \begin{smallmatrix} 1/2 & 0 \\ 0 & -1/2 \end{smallmatrix}\right) , \\[1ex]
\frac12 (a^\dag b + a b^\dag) &\longleftrightarrow \left( \begin{smallmatrix} 0 & 1/2 \\ 1/2 & 0 \end{smallmatrix}\right) ,\\
\frac{1}{2i} (a^\dag b - a b^\dag) &\longleftrightarrow \left( \begin{smallmatrix} 0 & -i/2 \\ i/2 & 0 \end{smallmatrix}\right) .
\end{align}
By exploiting these formulae, performing the computations for $2\times 2$ matrices, and bringing back the result with the Jordan map, it is possible to prove that~\cite[Appendix~5]{BARNETT-RADMORE}
\bb
U_\lambda = e^{\sqrt{\frac{1-\lambda}{\lambda}}\, a^\dag b}\, e^{-\frac12\ln \lambda \, (a^\dag a - b^\dag b)}\, e^{-\sqrt{\frac{1-\lambda}{\lambda}}\, a b^\dag}\, .
\label{splitting a bdag}
\ee
This decomposition can be employed to derive an expression for the output state obtained by mixing in a beam splitter of arbitrary transmissivity the vacuum $\ket{0}$ or the first Fock state $\ket{1}$ with another Fock states $\ket{n}$. Namely,
\begin{align}
U_\lambda \ket{0}\ket{n} &= \sum_{\ell=0}^{n} \sqrt{\binom{n}{\ell}} \left(1-\lambda\right)^{\frac{\ell}{2}} \lambda^{\frac{n-\ell}{2}} \ket{\ell}\ket{n-\ell}\, , \label{U 0n} \\
U_\lambda \ket{1}\ket{n} &= - \frac{1}{\sqrt{(n+1)(1-\lambda)}} \sum_{\ell=0}^{n+1} \sqrt{\binom{n+1}{\ell}} \left(1-\lambda\right)^{\frac{\ell}{2}} \lambda^{\frac{n-\ell}{2}} ((n+1)(1-\lambda) - \ell) \ket{\ell}\ket{n+1-\ell}\, . \label{U 1n}
\end{align}
To prove~\eqref{U 0n}, write
\begin{align*}
U_\lambda \ket{0}\ket{n} &= e^{\sqrt{\frac{1-\lambda}{\lambda}}\, a^\dag b}\, e^{-\frac12\ln\lambda \, (a^\dag a - b^\dag b)}\, e^{-\sqrt{\frac{1-\lambda}{\lambda}}\, a b^\dag} \ket{0}\ket{n} \\
&= e^{\sqrt{\frac{1-\lambda}{\lambda}}\, a^\dag b}\, e^{-\frac12\ln\lambda \, (a^\dag a - b^\dag b)} \ket{0}\ket{n} \\
&= e^{\sqrt{\frac{1-\lambda}{\lambda}}\, a^\dag b}\, \lambda^{\frac{n}{2}} \ket{0}\ket{n} \\
&= \lambda^{\frac{n}{2}} \sum_{\ell=0}^n \frac{1}{\ell!} \left( \frac{1-\lambda}{\lambda} \right)^{\frac{\ell}{2}} \left(\sqrt{\ell!}\, \ket{\ell}\right) \left( \sqrt{\frac{n!}{(n-\ell)!}}\, \ket{n-\ell} \right) \\
&= \sum_{\ell=0}^{n} \sqrt{\binom{n}{\ell}} \left(1-\lambda\right)^{\frac{\ell}{2}} \lambda^{\frac{n-\ell}{2}} \ket{\ell}\ket{n-\ell}\, .
\end{align*}
In the same spirit, one can compute
\begin{align*}
U_\lambda \ket{1}\ket{n} &= e^{\sqrt{\frac{1-\lambda}{\lambda}}\, a^\dag b}\, e^{-\frac12\ln\lambda \, (a^\dag a - b^\dag b)}\, e^{-\sqrt{\frac{1-\lambda}{\lambda}}\, a b^\dag} \ket{1}\ket{n} \\
&= e^{\sqrt{\frac{1-\lambda}{\lambda}}\, a^\dag b}\, e^{-\frac12\ln \lambda \, (a^\dag a - b^\dag b)} \left( \ket{1}\ket{n} - \sqrt{\frac{1-\lambda}{\lambda}} \sqrt{n+1} \ket{0}\ket{n+1} \right) \\
&= e^{\sqrt{\frac{1-\lambda}{\lambda}}\, a^\dag b} \left( \lambda^{\frac{n-1}{2}} \ket{1}\ket{n} - \sqrt{\frac{1-\lambda}{\lambda}} \sqrt{n+1} \lambda^{\frac{n+1}{2}} \ket{0}\ket{n+1} \right) \\
&= \lambda^{\frac{n-1}{2}} e^{\sqrt{\frac{1-\lambda}{\lambda}}\, a^\dag b} \ket{1}\ket{n} - \lambda^{\frac{n}{2}} \sqrt{(1-\lambda)(n+1)} e^{\sqrt{\frac{1-\lambda}{\lambda}}\, a^\dag b} \ket{0}\ket{n+1} \\
&= \lambda^{\frac{n-1}{2}} \sum_{\ell=0}^n \frac{1}{\ell!} \left(\frac{1-\lambda}{\lambda}\right)^{\frac{\ell}{2}}  \left( \sqrt{(\ell+1)!}\, \ket{\ell+1} \right) \left( \sqrt{\frac{n!}{(n-\ell)!}}\, \ket{n-\ell} \right) \\
&\quad - \lambda^{\frac{n}{2}} \sqrt{(1-\lambda)(n+1)} \sum_{\ell=0}^{n+1} \frac{1}{\ell!} \left(\frac{1-\lambda}{\lambda}\right)^{\frac{\ell}{2}} \left( \sqrt{\ell!}\, \ket{\ell} \right) \left( \sqrt{\frac{(n+1)!}{(n+1-\ell)!}}\, \ket{n+1-\ell} \right) \\
&= \frac{1}{\sqrt{n+1}} \sum_{\ell=1}^{n+1} \sqrt{\binom{n+1}{\ell}} (1-\lambda)^{\frac{\ell}{2}} \lambda^{\frac{n-\ell}{2}}\, \ell\, \ket{\ell}\ket{n+1-\ell} \\
&\quad - \sqrt{(n+1)(1-\lambda)} \sum_{\ell=0}^{n+1} \sqrt{\binom{n+1}{\ell}} (1-\lambda)^{\frac{\ell}{2}} \lambda^{\frac{n-\ell}{2}} \ket{\ell}\ket{n+1-\ell} \\
&= \frac{1}{\sqrt{(n+1)(1-\lambda)}} \sum_{\ell=0}^{n+1} \sqrt{\binom{n+1}{\ell}} \left(1-\lambda\right)^{\frac{\ell}{2}} \lambda^{\frac{n-\ell}{2}} (\ell-(n+1)(1-\lambda)) \ket{\ell}\ket{n+1-\ell}\, ,
\end{align*}
which proves~\eqref{U 1n}.

If instead of~\eqref{splitting a bdag} one employs the alternative decomposition
\bb
U_\lambda = e^{-\sqrt{\frac{1-\lambda}{\lambda}}\, ab^\dag}\, e^{\frac12\ln\lambda \, (a^\dag a - b^\dag b)}\, e^{\sqrt{\frac{1-\lambda}{\lambda}}\, a^\dag b}\, ,
\label{splitting adag b}
\ee
one finds that
\begin{align}
U_\lambda \ket{n}\ket{0} &= \sum_{\ell=0}^{n} (-1)^\ell \sqrt{\binom{n}{\ell}} \left(1-\lambda\right)^{\frac{\ell}{2}} \lambda^{\frac{n-\ell}{2}} \ket{n-\ell}\ket{\ell}\, , \label{U n0} \\
U_\lambda \ket{n}\ket{1} &= - \frac{1}{\sqrt{(n+1)(1-\lambda)}} \sum_{\ell=0}^{n+1} (-1)^\ell \sqrt{\binom{n+1}{\ell}} \left(1-\lambda\right)^{\frac{\ell}{2}}  \lambda^{\frac{n-\ell}{2}} ((n+1)(1-\lambda)- \ell) \ket{n+1-\ell}\ket{\ell}\, . \label{U n1}
\end{align}
Note that~\eqref{U n0} and~\eqref{U n1} can also be derived from~\eqref{U 0n} and~\eqref{U 1n} by applying the swap operator to both sides of the equations.

Finally, for future convenience we report the expressions of the matrices that represent $U_\lambda$ on subspaces with low total photon number. By applying~\eqref{splitting a bdag} or~\eqref{splitting adag b} one can verify that
\begin{align}
U_\lambda\big|_{\Span\{\ket{0}\ket{1},\,\ket{1}\ket{0}\}} &= \begin{pmatrix} \sqrt{\lambda} & -\sqrt{1-\lambda} \\ \sqrt{1-\lambda} & \sqrt{\lambda} \end{pmatrix} , \label{U 1} \\[2ex]
U_\lambda \big|_{\Span\{\ket{0}\ket{2},\,\ket{1}\ket{1},\,\ket{2}\ket{0}\}} &= \begin{pmatrix} \lambda & - \sqrt{2\lambda(1-\lambda)}  & 1-\lambda \\[1ex] \sqrt{2\lambda(1-\lambda)} & 2\lambda-1 & -\sqrt{2\lambda(1-\lambda)} \\[1ex] 1-\lambda & \sqrt{2\lambda(1-\lambda)}  & \lambda \end{pmatrix} . \label{U 2}
\end{align}

\subsection{General attenuators}

The family of channels that we consider here is that of \emph{general attenuators}~\cite{Koenig2015, KK-VV, Jack2018, Lim2019}, sometimes called additive noise channels~\cite{Koenig2015}. They are parametrised by a generic $m$-mode quantum state $\sigma$ and by a value of the associated transmissivity $0\leq \lambda\leq 1$. As reported in the main text~\eqref{Phi}--\eqref{Phi chi}, the action of a general attenuator $\Phi_{\lambda,\sigma}$ on a system $B$ is defined by $\Phi_{\lambda, \sigma}^B (\rho_B) \coloneqq \Tr_E\left[ U_\lambda^{BE} (\rho_B\otimes \sigma_E) \left(U_\lambda^{BE}\right)^\dag \right]$, which -- dropping the system labels for simplicity -- translates to $\chi_{\Phi_{\lambda, \sigma}(\rho)}(\alpha) = \chi_\rho\left(\sqrt\lambda\, \alpha\right) \chi_\sigma\left(\sqrt{1-\lambda}\, \alpha\right)$ at the level of characteristic functions. This particularly simple expression can be used in conjunction with the composition rule for displacement operators to prove the covariance formulae
\begin{align}
\Phi_{\lambda,\sigma} \circ \pazocal{D}_z &= \pazocal{D}_{\sqrt\lambda\, z}\circ \Phi_{\lambda,\sigma} \, , \label{covariance 1} \\
\Phi_{\lambda,\, \pazocal{D}_z(\sigma)} &= \pazocal{D}_{\sqrt{1-\lambda}\, z}\circ \Phi_{\lambda,\sigma}\, , \label{covariance 2}
\end{align}
where the displacement channel is defined as
\bb
\pazocal{D}_z(\rho)\coloneqq D(z)\,\rho\, D(z)^\dag\, .
\label{displacement channel}
\ee
Note that the identity~\eqref{covariance 2} has been used in the proof of Theorem~\ref{symmetric sigma thm}.

The canonical example of a general attenuator channel $\Phi_{\lambda,\sigma}$ -- say, in the single-mode case -- is obtained by setting $\sigma = \tau_\nu$, where the thermal state with mean photon number $
\nu$ is defined by~\eqref{thermal}. The resulting map $\mathcal{E}_{\lambda,\nu}\coloneqq \Phi_{\lambda,\, \tau_\nu}$ is usually referred to as a \emph{thermal attenuator}. An even simpler yet extremely important channel, called the \emph{quantum-limited attenuator} (or the \emph{pure loss channel}) and usually denoted with $\mathcal{E}_\lambda \coloneqq \mathcal{E}_{\lambda, 0} = \Phi_{\lambda,\, \ket{0}\!\bra{0}}$, is obtained by setting the temperature of the environment equal to zero.

The energy-constrained quantum capacity of the pure loss channel has been determined exactly. It reads
\bb
Q\left(\mathcal{E}_{\lambda}, N \right) = \max\left\{ g(\lambda N) - g\left( (1-\lambda) N\right),\, 0\right\} .
\label{pure loss Q}
\ee
The decisive step towards establishing~\eqref{pure loss Q} has been done by Wolf et al.~\cite[Eq.~(12)]{Wolf2007}, who proved that for this particular channel the regularisation in~\eqref{energy-constrained Q 1} is not needed. This implies that the quantum capacity is simply given by the coherent information~\eqref{energy-constrained Q 2}, which had been previously computed by Holevo and Werner~\cite[Eq.~(5.9)]{holwer}. A more complete discussion of these latter calculations, and in particular of why it suffices to consider thermal states at the input, can be found in Holevo's monograph~\cite[Propositions~12.38 and~12.47]{HOLEVO-CHANNELS} (see also the more recent version~\cite[Propositions~12.40 and~12.62]{HOLEVO-CHANNELS-2}).
The problem of completeness of the original argument was recently raised by Wilde and Qi~\cite[Remark~4]{Mark-energy-constrained}, and further elaborated on by Noh et al.~\cite[Theorem~9]{Noh2019}. An alternative derivation of the formula~\eqref{pure loss Q} has been put forward by Wilde et al.~\cite{Mark2012}.

We do not yet have an exact expression for the energy-constrained capacity of all thermal attenuators. However, many upper~\cite{PLOB, Rosati2018, Sharma2018, Noh2019} as well as lower~\cite{holwer, Noh2020} bounds have been discovered so far. We do not report the corresponding formulae here, as we do not need them. What we will need, instead, is a much simpler observation due to Caruso and Giovannetti~\cite{Caruso2006}.

\begin{lemma}[{\cite[p.~3]{Caruso2006}}] \label{Caruso V lemma}
For all $0\leq \lambda\leq \frac12$ and all $\nu\geq 0$, the thermal attenuator $\mathcal{E}_{\lambda, \nu}$ is anti-degradable, and thus $Q\left( \mathcal{E}_{\lambda, \nu}\right)=0$.
\end{lemma}

The above result can be further generalised thanks to the concept of channel $k$-extendibility. Here, anti-degradable channels are precisely those that are $2$-extendible. The complete characterisation of the $k$-extendibility regions of all thermal attenuators has been put forward recently~\cite{extendibility}.

We now turn to the problem of estimating the quantum capacity of general attenuators. We start by recalling the following elementary fact, that is part of the folklore.

\begin{lemma} \label{universal upper bound Q}
Let $\Phi$ be a quantum channel acting on a system of $m$ modes. For all $N\geq 0$, its energy-constrained quantum capacity satisfies
\bbb
Q\left(\Phi,\, N \right) \leq m\,g\left(\frac{N}{m}\right) ,
\eee
where the bosonic entropy is given by~\eqref{g}.
\end{lemma}

\begin{proof}
For any bipartite quantum system $AB$ we have that $H(AB)\geq \left|H(A) - H(B)\right|$. From this we deduce that the coherent information in~\eqref{energy-constrained Q 1} satisfies
\bbb
I_{\mathrm{coh}} (A\rangle B)_{\left( I_A\otimes \Phi_B\right)(\Psi_{AB})} = \left(H(B) - H(AB)\right)_{\left( I_A\otimes \Phi_B\right)(\Psi_{AB})} \leq H(A)_{\Psi_A} = S(\Psi_A) = S(\Psi_B) \leq m\, g\left( \frac{N}{m}\right) ,
\eee
where we used (i)~the fact that the initial state $\Psi_{AB} = \ketbra{\Psi}_{AB}$ is pure; and (ii)~the fact that the thermal state maximises the entropy for a given mean photon number, as stated in~\eqref{tau maximises entropy}. Since the above upper bound is additive, applying the LSD theorem~\eqref{energy-constrained Q 2} yields the claim.
\end{proof}

Exploiting known extremality properties of Gaussian states~\cite{Eisert-Wolf, Wolf2006}, the recent solution of the minimum output entropy conjecture~\cite{Jack-constrained} (see also~\cite{Giovadd, Giovadd-CMP}), and the even more recently established conditional entropy power inequality~\cite{Koenig2015, Jack2018}, Lim et al.~\cite{Lim2019} were able to prove the following more sophisticated bounds.

\begin{lemma}[{\cite[Sections~III and~IV]{Lim2019}}] \label{Lim lemma}
Let $\sigma$ be a single-mode state with mean photon number $\nu_\sigma$ and entropy $S(\sigma)$. Then, for all $0\leq \lambda\leq 1$ and $N\geq 0$ the energy-constrained quantum capacity of the corresponding general attenuator satisfies
\bb
g\left( (1\!-\!\lambda) g^{-1}\!( S(\sigma) ) + \lambda N \right) - S(\sigma) - g\left( \lambda \nu_\sigma + (1-\lambda)N \right) \leq Q\left( \Phi_{\lambda, \sigma}, N\right) \leq g\left( \lambda N + (1\!-\!\lambda) \nu_\sigma \right) - \ln \left( \lambda + (1\!-\!\lambda) e^{S(\sigma)} \right) ,
\label{bounds Lim}
\ee
where $g^{-1}$ is the inverse function of the bosonic entropy $g$ defined by~\eqref{g}.
\end{lemma}

\begin{rem}
Note that the lower bound in~\eqref{bounds Lim} always vanishes when $\lambda\leq 1/2$. Indeed, using the subadditivity~\eqref{subadditivity g} and monotonicity of the bosonic entropy yields
\begin{align*}
&g\left( (1-\lambda) g^{-1}( S(\sigma) ) + \lambda N \right) - S(\sigma) - g\left( \lambda \nu_\sigma + (1-\lambda)N \right) \\
&\qquad \leq g\left( (1-\lambda) g^{-1}( S(\sigma) )\right) + g\left(\lambda N \right) - S(\sigma) - g\left( \lambda \nu_\sigma + (1-\lambda)N \right) \\
&\qquad \leq g\left( g^{-1}( S(\sigma) )\right) + g\left(\lambda N \right) - S(\sigma) - g\left( (1-\lambda) N \right) \\
&\qquad = g\left(\lambda N \right) - g\left( (1-\lambda) N \right) \\
&\qquad \leq 0\, ,
\end{align*}
where the last inequality holds provided that $\lambda\leq 1/2$, again using the monotonicity of $g$. It follows that the recent results by Lim et al. ~\cite{Lim2019} cannot be possibly used to detect a positive quantum capacity below the threshold value $\lambda=1/2$.
\end{rem}

\begin{rem}
The upper bound in~\eqref{bounds Lim} diverges for every fixed $N$ and $\lambda>0$ when $\nu_\sigma\to\infty$. However, we have already seen in Lemma~\ref{universal upper bound Q} that the maximum capacity $Q\left(\Phi_{\lambda, \sigma},\, N\right)$ stays finite in the same limit.
\end{rem}

\subsection{Quantum capacity and entanglement distribution}

Before we move on, let us briefly comment on the problem of entanglement distribution with general attenuators. It is well known that the channel $\mathcal{E}_{\lambda, \nu}$ is entanglement-breaking if and only if $0\leq \lambda\leq \frac{\nu}{\nu+1}$~\cite[Eq.~(38)]{Holevo-EB} (see also~\cite{extendibility} for the generalisation to $k$-extendibility). This in particular implies that the pure loss channel, corresponding to the case $\nu=0$, is not entanglement-breaking for any $\lambda>0$. In other words, it is possible to distribute entanglement using the pure loss channel for arbitrary small non-zero values of the transmissivity. This can be easily verified with a one-line computation by sending one half of a Bell state $\ket{\Psi}_{AB}=\frac{1}{\sqrt{2}} \left( \ket{01}+\ket{10}\right)_{AB}$ through the channel $\mathcal{E}_\lambda^{B}$. The result is
\bb
\left(I^A \otimes \mathcal{E}_\lambda^{B} \right)(\Psi_{AB}) = \frac{1-\lambda}{2} \ketbra{00}_{AB} + \frac{\lambda}{2} \ketbra{01}_{AB} + \frac12 \ketbra{10}_{AB} + \frac{\sqrt\lambda}{2} \left( \ketbraa{01}{10} + \ketbraa{10}{01}\right)_{AB}\, .
\label{simple-transmission}
\ee
This is effectively a two-qubit state, and the fact that it is indeed entangled for all values of $\lambda>0$ can be straightforwardly verified by an application of the partial transposition criterion~\cite{PeresPPT}. Since all entangled two-qubit states are distillable~\cite{2-qubit-distillation, Horodecki-review}, we conclude that the state in~\eqref{simple-transmission} is also distillable. 
A result by Masanes~\cite{Masanes2006} then guarantees that a large number of copies of~\eqref{simple-transmission} can be processed locally -- without communication -- so as to yield a \emph{non-local state}, that is, a state that violates a Bell inequality~\cite{Brunner-review}. Incidentally, the Bell inequality in question can be taken to be the Clauser--Horne--Shimony--Holt inequality~\cite{CHSH}. Therefore, the pure loss channel $\mathcal{E}_\lambda$ can be used (in conjunction with local processing) to construct states that violate a Bell inequality for all $\lambda>0$.

At first glance, this may appear to contradict the above Lemma~\ref{Caruso V lemma} \cite{Caruso2006}. The contradiction is resolved once one observes that sending quantum messages through a channel is a different -- and indeed harder -- task than that of using it to distribute a non-local state. While any channel with positive quantum capacity can be used to create maximally entangled (and hence non-local) states between sender and receiver, the above example shows that there exist channels with vanishing quantum capacity that can anyway distribute non-local states. What makes our main result (Theorem~\ref{Behemoth thm}) non-trivial is that while it is relatively easy to show that general attenuators with arbitrary low non-zero values of the transmissivity are fit to the latter task, this does not imply anything about the former (and harder) one.

\section{Convex combinations of Gaussian states}

Throughout this section, we look at general attenuators whose environment state is a convex combination of Gaussian states. Note that this family of states encompasses the so-called \emph{classical states}~\cite{Bach1986, Yadin2018}, which by definition can be written as convex combinations of coherent states, i.e.
\bb
\sigma = \int d\mu(\alpha)\, \ketbra{\alpha}\, ,
\ee
where $\mu$ is a probability measure on $\C^m$.

We start by showing how to apply the data processing bound to constrain the quantum capacity of general attenuators. 

\begin{lemma} \label{data processing attenuators lemma}
Let $0\leq \lambda,\mu\leq 1$, and let $\sigma,\omega$ be $m$-mode states. Then we have the composition rule
\bb
\Phi_{\lambda,\sigma}\circ \Phi_{\mu,\omega} = \Phi_{\lambda\mu,\,\tau}\, ,
\label{composition 1}
\ee
where
\bb
\tau \coloneqq \Phi_{\frac{\lambda(1-\mu)}{1-\lambda\mu},\, \sigma} (\omega)= \Phi_{\frac{1-\lambda}{1-\lambda \mu},\, \omega}(\sigma)\, .
\label{composition 2}
\ee
\end{lemma}

\begin{proof}
The easiest way to verify~\eqref{composition 1} is by looking at the transformation rules for characteristic functions. For an arbitrary input state $\rho$, using~\eqref{Phi chi} multiple times we obtain that
\begin{align*}
\chi_{\left(\Phi_{\lambda,\sigma} \circ \Phi_{\mu,\omega}\right)(\rho)}(\alpha) &= \chi_{\Phi_{\mu,\omega}(\rho)}\left( \sqrt\lambda\, \alpha \right) \chi_{\sigma}\left(\sqrt{1-\lambda}\, \alpha\right) \\
&= \chi_{\rho}\left(\sqrt{\lambda\mu}\, \alpha\right) \chi_{\omega}\left(\sqrt{\lambda\left(1-\mu\right)}\,\alpha\right) \chi_{\sigma}\left(\sqrt{1-\lambda}\, \alpha\right) \\
&= \chi_{\rho}\left(\sqrt{\lambda\mu}\, \alpha\right) \chi_{\eta}\left(\sqrt{1-\lambda\mu}\,\alpha\right) \\
&= \chi_{\Phi_{\lambda\mu,\,\eta}(\rho)}(\alpha)\, .
\end{align*}
Since quantum states are in one-to-one correspondence with characteristic functions, this implies that $\Phi_{\lambda\mu,\,\eta}(\rho) = \left(\Phi_{\lambda,\sigma} \circ \Phi_{\mu,\omega}\right)(\rho)$. Given that $\rho$ was arbitrary, the proof is complete.
\end{proof}

A first immediate corollary is as follows.

\begin{cor} \label{anti-deg composition cor}
Let $0\leq \lambda,\mu\leq 1$, and let $\sigma,\omega$ be $m$-mode states. Define $\tau$ as in~\eqref{composition 2}. Then: (a)~if $\Phi_{\mu,\omega}$ is anti-degradable, then so is $\Phi_{\lambda\mu,\, \tau}$; (b)~it holds that
\bb
Q\left( \Phi_{\lambda\mu,\,\tau} \right) \leq \min\left\{ \Phi_{\lambda,\sigma},\, \Phi_{\mu,\omega} \right\} .
\label{composition capacity}
\ee 
\end{cor}

\begin{proof}
Claim~(a) is a consequence of the fact that set of anti-degradable channels is invariant by post-processing. Claim~(b), instead, follows from the observation that in the definition of quantum capacity any pre- or post-processing can be included into the encoding or decoding transformations.
\end{proof}

We now show that the phenomenon illustrated in Theorem~\ref{Behemoth thm} does not occur for general attenuators whose environment state is a convex combination of Gaussian states.

\begin{cor} \label{anti-deg convex hull G cor}
Let $\sigma$ be a state in the convex hull of all Gaussian states. Then $\Phi_{\lambda, \sigma}$ is anti-degradable for all $0\leq \lambda\leq 1/2$, and in particular
\bb
Q\left( \Phi_{\lambda, \sigma} \right) \equiv 0\qquad \forall\ 0\leq \lambda\leq \frac12\, .
\ee
\end{cor}

\begin{proof}
For $\sigma$ satisfying the hypothesis, we have that $\Phi_{\lambda, \sigma}$ is a convex combination of channels of the form $\Phi_{\lambda,\, \sigma_\G}$, where $\sigma_\G$ is Gaussian. Since the set of anti-degradable channels is convex~\cite[Appendix~A.2]{Cubitt2008}, it suffices to prove that $\Phi_{\lambda,\, \sigma_\G}$ is anti-degradable for all $0\leq \lambda\leq 1/2$.

Moreover, we can assume without loss of generality that $\sigma_\G$ is centred, i.e.\ that $\Tr[\sigma_\G\, a_j]\equiv 0$ for all $j=1,\ldots,m$, where $a_j$ are the annihilation operators. In fact, $\sigma_\G$ can always be displaced by an arbitrary amount by means of a unitary post-processing as in~\eqref{covariance 2}. Note that unitary post-processing does not affect anti-degradability, and that $D(z) a_j D(z)^\dag = a_j - z_j$. Thus, we can make sure that $\pazocal{D}_z\left(\sigma_\G\right)$ is centred by choosing $z$ appropriately.

In light of the above reasoning, from now on we shall assume that $\sigma_\G$ is centred. The characteristic function of $\sigma_\G$ then is a centred Gaussian, entailing that
\bbb
\chi_{\sigma_\G}\left( \sqrt{\eta}\, \alpha\right) \chi_{\sigma_\G}\left( \sqrt{1-\eta}\, \alpha\right) \equiv \chi_{\sigma_\G} \left( \alpha \right) \qquad \forall\ \alpha\in \C^m,\quad \forall\ 0\leq \eta\leq 1\, .
\eee
Using~\eqref{Phi chi}, this translates to
\bb
\Phi_{\eta,\,\sigma_\G}(\sigma_\G) \equiv \sigma_\G\qquad \forall\ 0\leq \eta\leq 1\, .
\label{stability G}
\ee
Leveraging~\eqref{composition 1}--\eqref{composition 2},  we see that
\bbb
\Phi_{\lambda,\, \sigma_\G} = \Phi_{2\lambda,\, \sigma_\G} \circ \Phi_{1/2,\, \sigma_\G}
\eee
for all $0\leq \lambda\leq 1/2$. Note that $\Phi_{1/2,\, \sigma_\G}$ is anti-degradable by Theorem~\ref{symmetric sigma thm}. Since anti-degradable channels remain such upon post-processing~\cite[Lemma~17]{Cubitt2008}, we conclude that also $\Phi_{\lambda,\, \sigma_\G}$ is anti-degradable, completing the proof.
%The unconstrained quantum capacity is defined as $Q\left(\Phi_{\lambda, \sigma}\right) = \lim_{N\to\infty} Q\left( \Phi_{\lambda, \sigma}, N\right)$, where the right-hand side is given by~\eqref{energy-constrained Q 1}--\eqref{energy-constrained Q 2}. Let us now argue that $Q\left( \Phi_{\lambda, \sigma}, N\right)$ is continuous in $\sigma$ with respect to the trace distance for every fixed $N$. (i)~By a result of Becker et at.~\cite[Lemma~23]{QCLT} we have that $\left\|\Phi_{\lambda, \sigma_1} - \Phi_{\lambda, \sigma_2}\right\|_\diamond \leq \left\|\sigma_1-\sigma_2\right\|_1$, where the left-hand side is the diamond distance between channels, in turn defined by~\cite{Aharonov1998}
%\bbb \left\| \mathcal{L}\right\|_\diamond \coloneqq \sup_{\ket{\Psi}} \left\|\left( \mathcal{L}\otimes I\right)\left(\Psi\right) \right\|_1\, . \eee
%(ii) The diamond distance upper bounds the 
\end{proof}

Another consequence of Lemma~\ref{data processing attenuators lemma} is that the quantum capacity of a general attenuator is monotonically increasing as a function of the transmissivity for a fixed Gaussian environment state. By comparison, remember that in the main text we have instead shown that monotonicity fails to hold when the environment state is a Fock state.

\begin{cor}
Let $\sigma=\sigma_\G$ be an arbitrary $m$-mode Gaussian state. Then the function
\bb
\lambda\longmapsto Q \left(\Phi_{\lambda,\,\sigma_\G}\right)
\ee
is monotonically increasing for all $0\leq \lambda\leq 1$, and strictly zero for $0\leq \lambda\leq 1/2$.
\end{cor}

\begin{proof}
The proof is along the same lines as that of Corollary~\ref{anti-deg convex hull G cor}. We can assume without loss of generality that $\sigma_\G$ is centred, which in turn implies that~\eqref{stability G} holds. Picking $0\leq \lambda'\leq \lambda\leq 1$ and setting $\mu\coloneqq \frac{\lambda'}{\lambda}$ and $\eta\coloneqq \frac{1-\lambda}{1-\lambda\mu} = \frac{1-\lambda}{1-\lambda'}$ in~\eqref{composition 2}--\eqref{composition capacity}, we deduce that
\bbb
\Phi_{\lambda',\, \sigma_\G} = \Phi_{\lambda\mu,\, \sigma_\G} = \Phi_{\lambda\mu,\, \Phi_{\eta,\,\sigma_\G}(\sigma_\G)} = \Phi_{\lambda,\sigma_\G} \circ \Phi_{\mu,\, \sigma_\G} \, .
\eee
Then, applying~\eqref{composition capacity} we conclude that $Q\left( \Phi_{\lambda',\, \sigma_\G} \right) \leq Q\left( \Phi_{\lambda,\sigma_\G}\right)$, completing the proof.
\end{proof}

\section{Positive capacity at $\boldsymbol{\lambda=1/2}$}

Given the fact that general attenuator channels of the form $\Phi_{1/2,\, \sigma}$ always output states with positive Wigner functions~\cite[Lemma~16]{QCLT}, one may be tempted to conjecture that their quantum capacities vanish. Interestingly, this is not the case, as the next example shows.

\begin{ex} \label{positive capacity at 1/2 ex}
For $0\leq \eta\leq 1$, set $\xi(\eta)=\ketbra{\xi(\eta)}$, with 
\bb
\ket{\xi(\eta)} \coloneqq \sqrt{\eta}\ket{0} - \sqrt{1-\eta} \ket{1}\, .
\label{xi eta}
\ee
We will see (numerically) that for all $\eta\in (0,1)$ the channel $\Phi_{1/2,\,\xi(\eta)}$ has nonzero quantum capacity. A lower bound on $Q\left( \Phi_{1/2,\,\xi(\eta)} \right)$ is plotted in Figure~\ref{Icoh figure}.

\smallskip
To estimate the quantum capacity of $\Phi_{1/2,\,\xi(\eta)}$ from below, we apply the achievability part of the LSD theorem. This is done by finding a suitable ansatz for the state $\ket{\Psi}_{AB}$ to be plugged into~\eqref{energy-constrained Q 2}. Let us define the family of two-mode states
\bb
\ket{\Psi(\eta)}_{AB} \coloneqq \sqrt{\eta (1-\eta)}\ket{0}_A\ket{0}_B + (1-\eta)\ket{0}_A \ket{1}_B+\sqrt{\eta}\ket{1}_A \ket{0}_B\, .
\ee
Upon re-ordering the terms, the joint state reads
\bb
\begin{aligned}
\ket{\Psi(\eta)}_{AB} \ket{\xi(\eta)}_E &= \eta\left( \ket{1}_A+\sqrt{1-\eta} \ket{0}_A \right) \ket{0}_B \ket{0}_E - \sqrt{\eta(1-\eta)}\left( \sqrt{1-\eta} \ket{0}_A + \ket{1}_A\right) \ket{0}_B \ket{1}_E \\
&\qquad +\sqrt{\eta} (1-\eta)\ket{0}_A \ket{1}_B \ket{0}_E - (1-\eta)^{3/2} \ket{0}_A \ket{1}_B \ket{1}_E\, .
\end{aligned}
\ee
Using the explicit representations~\eqref{U 1}--\eqref{U 2} of the action of the beam splitter unitary on the low photon number subspaces, it is not difficult to see that the tripartite output state, which we denote as
\bb
\ket{\zeta(\lambda, \eta)}_{A BE} \coloneqq U^{BE}_{\lambda}\, \ket{\Psi(\eta)}_{AB} \ket{\xi(\eta)}_E\, ,
\label{zeta}
\ee
reads
\begin{align*}
\ket{\zeta(\lambda, \eta)}_{\!A\!B\!E} =&\ \eta\left( \ket{1}_A +\sqrt{1-\eta} \ket{0}_A \right) \ket{0}_B \ket{0}_E - \sqrt{\eta(1-\eta)}\left( \sqrt{1-\eta} \ket{0}_A + \ket{1}_A \right) \! \left( \sqrt\lambda \ket{0}_B \ket{1}_E + \sqrt{1-\lambda} \ket{1}_B \ket{0}_E\right) \\
&\quad +\sqrt{\eta} (1-\eta)\ket{0}_A \left( -\sqrt{1-\lambda} \ket{0}_B \ket{1}_E + \sqrt\lambda \ket{1}_B \ket{0}_E \right) \\
&\quad - (1-\eta)^{3/2} \ket{0}_A \left( -\sqrt{2\lambda(1-\lambda)} \ket{0}_B \ket{2}_E + (2\lambda-1)\ket{1}_B \ket{1}_E + \sqrt{2\lambda(1-\lambda)} \ket{2}_B\ket{0}_E\right) \\
=&\ \bigg( \eta\sqrt{1-\eta} \ket{0}_A\ket{0}_B + \sqrt{\eta} (1-\eta) \left( \sqrt\lambda - \sqrt{1-\lambda} \right) \ket{0}_A \ket{1}_B \\
&\quad + \eta \ket{1}_A \ket{0}_B - \sqrt{\eta(1-\eta)(1-\lambda)} \ket{1}_A \ket{1}_B - (1-\eta)^{3/2}\sqrt{1-\lambda} \ket{0}_A \ket{2}_B \bigg) \ket{0}_E \\
&\quad - \left( \sqrt{\eta} (1-\eta) \left(\sqrt\lambda + \sqrt{1-\lambda}\right) \ket{0}_A \ket{0}_B + (1-\eta)^{3/2} (2\lambda -1) \ket{0}_A \ket{1}_B + \sqrt{\eta (1-\eta)\lambda} \ket{1}_A \ket{0}_B \right) \ket{1}_E \\
&\quad + (1-\eta)^{3/2} \sqrt{2\lambda (1-\lambda)} \ket{0}_A \ket{0}_B \ket{2}_E\, . 
\end{align*}
From now on, we focus only on the case $\lambda=1/2$. Upon tedious yet straightforward calculations, we find that with respect to the lexicographically ordered product basis $\{\ket{0}_A,\ket{1}_A\}\otimes \{\ket{0}_B,\ket{1}_B,\ket{2}_B\}$ we have that
\bb
\zeta_{AB}(1/2,\, \eta) = \left(I^A\otimes \Phi_{1/2,\,\xi(\eta)}^B \right)\left(\Psi_{AB}(\eta)\right) = \begin{pmatrix}
 \frac{1}{2} \left(1+\eta-3 \eta^2+\eta^3\right) & 0 & -\frac{(1-\eta)^2 \eta}{\sqrt{2}} & \eta\sqrt{1-\eta} & -\frac{(1-\eta) \eta^{3/2}}{\sqrt{2}} & 0 \\
 0 & 0 & 0 & 0 & 0 & 0 \\
 -\frac{(1-\eta)^2 \eta}{\sqrt{2}} & 0 & \frac{1}{2} (1-\eta)^3 & -\frac{(1-\eta)^{3/2} \eta}{\sqrt{2}} & \frac{1}{2} (1-\eta)^2 \sqrt{\eta} & 0 \\
 \eta\sqrt{1-\eta} & 0 & -\frac{(1-\eta)^{3/2} \eta}{\sqrt{2}} & \frac{1}{2} \eta (1+\eta) & -\frac{\eta^{3/2} \sqrt{1-\eta}}{\sqrt{2}} & 0 \\
 -\frac{(1-\eta) \eta^{3/2}}{\sqrt{2}} & 0 & \frac{1}{2} (1-\eta)^2 \sqrt{\eta} & -\frac{\eta^{3/2} \sqrt{1-\eta}}{\sqrt{2}} & \frac{1}{2} (1-\eta) \eta & 0 \\
 0 & 0 & 0 & 0 & 0 & 0
\end{pmatrix}
\label{rho AB kappa}
\ee
and that
\bb
\zeta_{B}(1/2,\, \eta) = \Phi_{1/2,\,\xi(\eta)}^B \left(\Psi_{B}(\eta)\right) = \begin{pmatrix}
 \frac{1}{2} \left(1+2 \eta-2 \eta^2+\eta^3\right) & -\frac{\eta^{3/2} \sqrt{1-\eta}}{\sqrt{2}} & -\frac{(1-\eta)^2 \eta}{\sqrt{2}} \\
 -\frac{\eta^{3/2} \sqrt{1-\eta}}{\sqrt{2}} & \frac{1}{2} (1-\eta) \eta & 0 \\
 -\frac{(1-\eta)^2 \eta}{\sqrt{2}} & 0 & \frac{1}{2} (1-\eta)^3
\end{pmatrix} .
\label{rho B kappa}
\ee

We are now ready to apply the LSD theorem to our case. Note that the mean photon number of $\phi_B(\eta)$ is precisely $(1-\eta)^2$. Then, employing~\eqref{energy-constrained Q 1}--\eqref{energy-constrained Q 2} we find that
\bb
Q\left( \Phi_{1/2,\,\xi(\eta)} \right) \geq Q\left( \Phi_{1/2,\,\xi(\eta)},\, (1-\eta)^2\right) \geq I_{\mathrm{coh}}(A\rangle B)_{\zeta_{AB}(1/2,\, \eta)} .
\label{positive capacity 1/2}
\ee
The coherent information $I_{\mathrm{coh}}(A\rangle B)_{\zeta_{AB}(1/2,\, \eta)}$ is plotted in Figure~\ref{Icoh figure}. The numerics shows clearly that this is strictly positive for all $\eta\in (0,1)$. We do not provide an analytical proof of this claim, because it is not necessary for what follows. In our proof of Theorem~\ref{Behemoth thm} we will only use the easily verified fact that $I_{\mathrm{coh}}(A\rangle B)_{\zeta_{AB}(1/2,\, \eta)}>0$ for \emph{some} values of $\eta$.

\begin{figure}[ht]
\centering
\includegraphics{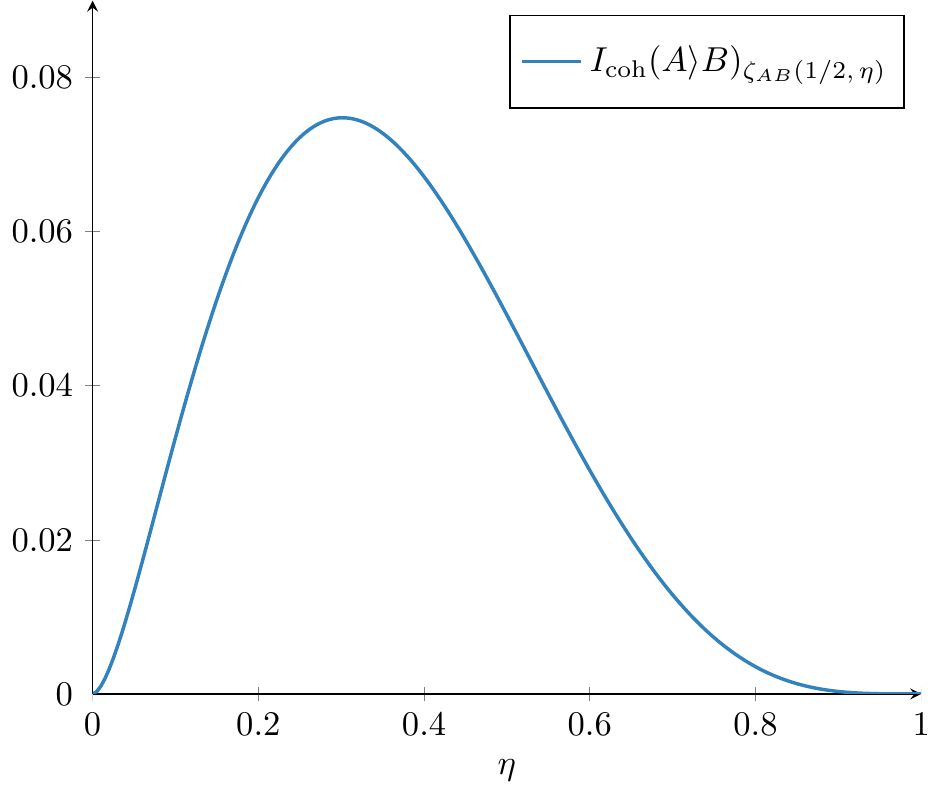}
\caption{The coherent information lower bound~\eqref{Icoh} on the quantum capacity of the channel $\Phi_{1/2,\, \xi(\eta)}$ defined by the environment state~\eqref{xi eta}. The maximum can be numerically evaluated, yielding $\max_{0\leq \eta\leq 1} I_{\mathrm{coh}}(A\rangle B)_{\zeta_{AB}(1/2,\, \eta)} \approx 0.0748$.}
\label{Icoh fig}
\end{figure}
\end{ex}

%\begin{rem}
%By optimising over a larger set of states one can detect an even larger coherent information. The largest lower bound on the quantum capacity that I have been able to determine reads
%\bb Q\left(\pazocal{N}_{0.29}\right) > 0.0762\, .\ee
%\end{rem}

Incidentally, general attenuators can have a substantially larger transmission capacity if one allows for a higher input power to be deployed.

\begin{ex} \label{Martins ex}
For $n\geq 3$ to be fixed, consider the environment state $\xi'(n) = \ketbra{\xi'(n)}$, with
\bb
\ket{\xi'(n)}\coloneqq \frac{\ket{n-1}+\ket{n}}{\sqrt2}\, .
\label{xi'}
\ee
We look at the transmission scheme identified by an initial state
\bb
\ket{\Psi'(n)}_{AB} \coloneqq \frac12 \left( \ket{0}_A\left( \ket{n-1}_B+\ket{n}_B\right) + \ket{1}_A \left( \ket{n-3}_B + \ket{n-2}_B\right) \right) .
\ee
Note that the mean photon number of $\Psi'_B(n)$ is $n-\frac32$. Applying the LSD theorem in the form of~\eqref{energy-constrained Q 1}--\eqref{energy-constrained Q 2} then yields
\bb
Q\left( \Phi_{1/2,\,\xi'(n)} \right) \geq Q\left( \Phi_{1/2,\,\xi'(n)},\, n-\frac32 \right) \geq I_{\mathrm{coh}}\left( A\rangle B\right)_{\zeta'_{AB}\left(1/2,\,n\right)} \, ,
\label{lower bound capacity Martin}
\ee
with $\zeta'_{AB}(1/2,\,n)\coloneqq \left( I^A\otimes \Phi_{1/2,\, \xi'(n)}\right) \left( \Psi'_{AB}(n)\right)$. The values of the right-hand side of~\eqref{lower bound capacity Martin} for $n=3,\ldots, 35$ are reported in Figure~\ref{Martins plot}. For $n=54$ the lower bound evaluates to around $0.3530$.

\begin{figure}[ht]
\centering
\includegraphics{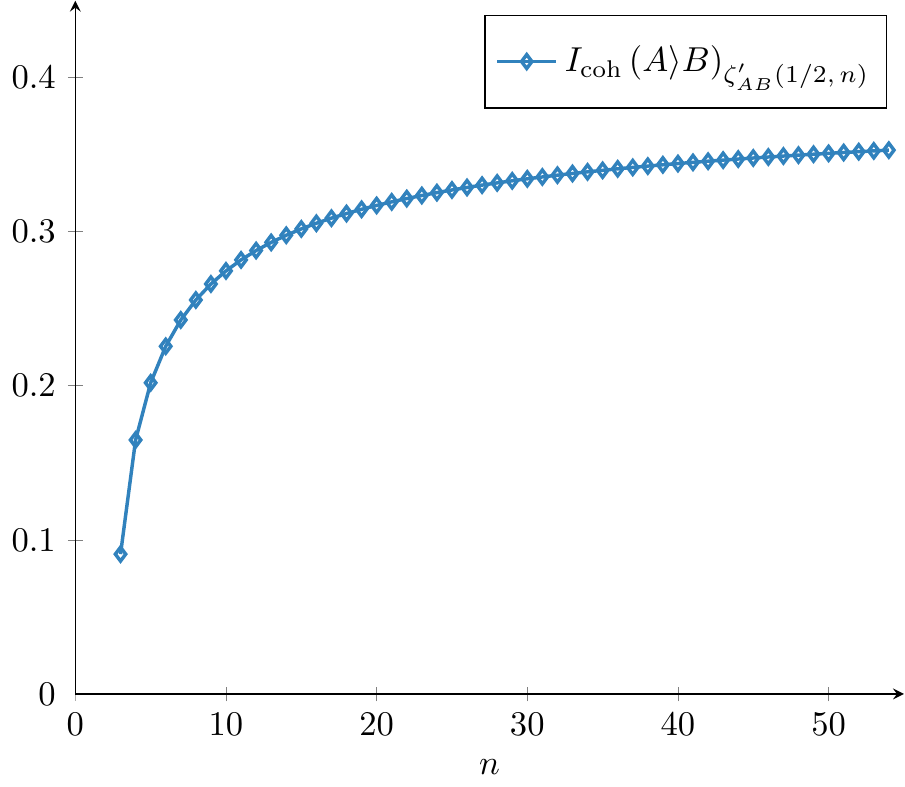}
\caption{The coherent information lower bound~\eqref{lower bound capacity Martin} on the quantum capacity of the channel $\Phi_{1/2,\, \xi'(n)}$ defined by the environment state in~\eqref{xi'}. For $n=54$ we obtain that $I_{\mathrm{coh}}\left( A\rangle B\right)_{\zeta'_{AB}\left(1/2,\, 54\right)} \approx 0.3530$.}
\label{Martins plot}
\end{figure}
\end{ex}

\section{Positive capacity at arbitrary transmissivity}

\begin{manualthm}{\ref{Behemoth thm}}
For all $0<\lambda \leq 1$ there exists a single-mode pure state $\sigma(\lambda)$ such that 
\bb
Q\big(\Phi_{\lambda,\,\sigma(\lambda)}\big) \geq Q\big(\Phi_{\lambda,\,\sigma(\lambda)},\, 1/2 \big) \geq c
\ee
for some universal constant $c>0$. 
%Both $\sigma(\lambda)$ and $c$ are explicitly given in the proof.
For an arbitrary $\epsilon>0$, one can take $\sigma(\lambda)$ to be: (a)~the vacuum for $1/2+\epsilon\leq \lambda\leq 1$; (b)~the state $\xi(1/3)$ defined by~\eqref{xi eta} for $1/2-\epsilon\leq \lambda\leq 1/2+\epsilon$; (c)~the Fock state $\ketbra{2}$ for $1/3\leq \lambda\leq 1/2-\epsilon$; and (d)~the Fock state $\ketbra{n}$ for $1/(n+1)\leq \lambda \leq 1/n$, for all integers $n\geq 3$.
\end{manualthm}

\begin{rem}
The value of the constant $c$ in Theorem~\ref{Behemoth thm} will depend on $\epsilon$. For the numerical determination of the optimal value yielded by our argument, see Remark~\ref{value c rem}. We also examine a closely related question in the subsequent Remark~\ref{value c0 rem}.
\end{rem}

Throughout this section we will provide a complete proof of the above result. In light of its complexity, we will break it down into several elementary steps. Further considerations and some extensions of Theorem~\ref{Behemoth thm} are postponed to the end of this section. Here is a brief account of the content of the various subsections:
\begin{itemize}
\item[\ref{transmission scheme subsec}.] Here we fix a transmission scheme, that is, a family of environment states $\sigma(n)$ (for $n=1,2,\ldots$) and an ansatz $\ket{\Psi}_{AB}$ to be plugged into the coherent information~\eqref{energy-constrained Q 2}.
%(or into its energy-constrained version~\eqref{energy-constrained Q 2}).
The result will be a lower bound of the form $Q\left(\Phi_{\lambda,\, \sigma(\lambda)}\right)\geq H(p(n,\lambda))-H(q(n,\lambda))$, where $p(n,\lambda)$ and $q(n,\lambda)$ are appropriate probability distributions over some index $\ell\in \{0,\ldots, n+1\}$ (Proposition~\ref{Phi n Q estimate prop}).
\item[\ref{sorting q subsec}.] We will then proceed to identify a range of values of $\lambda$ (depending on $n$) for which $q(n,\lambda)$ can be sorted in ascending order by a fixed permutation (luckily enough, this turns out to be the identity). The result is contained in Proposition~\ref{q sorted prop}.
\item[\ref{sorting p subsec}.] The same is then done for $p(n,\lambda)$, with considerably more effort and by keeping three distinct possibilities on the table (Proposition~\ref{p sorted prop}).
\item[\ref{majorisation subsec}.] The crux of the argument is to verify that for a sufficiently large range of values of $\lambda$ (for varying $n$) the probability distribution $q(n,\lambda)$ actually \emph{majorises} $p(n,\lambda)$.
\item[\ref{concluding subsec}.] The existence of a majorisation relation between $q(n,\lambda)$ and $p(n,\lambda)$ allows us to exploit a beautiful inequality due to Ho and Verd\'u~\cite[Theorem~3]{Ho2010} to lower bound their entropy difference by means of the relative entropy distance. In turn, this can be lower bounded in terms of their total variation distance thanks to Pinsker's inequality (see also~\cite[p.~58]{CSISZAR-KOERNER} and references therein). We look at the resulting bounds and draw our conclusions.
\item[\ref{further considerations subsec}] As a small addition to the actual proof of Theorem~\ref{Behemoth thm} (that concludes itself in Section~\ref{concluding subsec}), we show that, if one excludes the two base cases $n=2,3$, the majorisation between $q(n,\lambda)$ and $p(n,\lambda)$ becomes of a very special type (Proposition~\ref{majorisation n>3 prop}).
%\item[\ref{extensions subsec}] Finally, we discuss an immediate extension of Theorem~\ref{Behemoth thm} to the case of multiple concatenated beam splitters. Some open questions are highlighted.
\end{itemize}

\subsection{A transmission scheme} \label{transmission scheme subsec}

\begin{prop} \label{Phi n Q estimate prop}
Set $\ket{\Psi}_{AB} \coloneqq \frac{1}{\sqrt2}\big( \ket{0}_A\ket{1}_B + \ket{1}_A \ket{0}_B\big)$ and $\sigma(n)\coloneqq \ketbra{n}$. Then for all $\lambda\in (0,1)$ it holds that
\bb
Q\left( \Phi_{\lambda,\, \sigma(n)}, 1/2\right) \geq \mathcal{I}(n,\lambda) \coloneqq \icoh_{\big(I^A\otimes \Phi^B_{\lambda,\, \sigma(n)}\big)(\Psi_{AB})} = H\left(p(n,\lambda)\right) - H\left(q(n,\lambda)\right)\, ,
\label{Phi n Q estimate}
\ee
where $H$ denotes the Shannon entropy, and the two probability distributions $p(n,\lambda)=\left( p_0(n,\lambda),\ldots, p_{n+1}(n,\lambda) \right)$ and $q(n,\lambda)=\left( q_0(n,\lambda),\ldots, q_{n+1}(n,\lambda) \right)$ are defined by
\begin{align}
p_\ell(n,\lambda) &\coloneqq \frac{1}{2(n+1)(1-\lambda)} \binom{n+1}{\ell} \left(1-\lambda\right)^{\ell}\lambda^{n-\ell} \left((1-\lambda)(n-\ell+ 1) + \left( (n+1)(1-\lambda) - \ell \right)^2\right) \label{p ell} \\[2ex]
q_\ell(n,\lambda) &\coloneqq \frac{1}{2(n+1)(1-\lambda)} \binom{n+1}{\ell} \left(1-\lambda\right)^{\ell}\lambda^{n-\ell} \left( \lambda\ell + \left( (n+1)(1-\lambda) - \ell \right)^2\right) \label{q ell}
\end{align}
\end{prop}

\begin{proof}
Thanks to~\eqref{U 0n} and~\eqref{U 1n}, the action of the beam splitter on the $BE$ system can be expressed as
\begin{align*}
& U_\lambda^{BE} \ket{\Psi}_{\!AB}\ket{n}_{\!E} \\
&\quad = - \frac{1}{\sqrt2} \ket{0}_{\!A} \left( \frac{1}{\sqrt{(n\!+\!1)(1\!-\!\lambda)}} \sum_{\ell=0}^{n+1} \sqrt{\!\binom{n\!+\!1}{\ell}\!} \left(1\!-\!\lambda\right)^{\frac{\ell}{2}} \lambda^{\frac{n-\ell}{2}} \left((n\!+\!1)(1\!-\!\lambda)-\ell\right) \ket{\ell}_{\!B} \! \ket{n\!+\!1\!-\!\ell}_{\!E} \right) \\
&\quad \quad + \frac{1}{\sqrt2} \ket{1}_{\!A} \left( \sum_{\ell=0}^{n} \sqrt{\!\binom{n}{\ell}\!} \left(1\!-\!\lambda\right)^{\frac{\ell}{2}} \lambda^{\frac{n\!-\!\ell}{2}} \ket{\ell}_B \ket{n\!-\!\ell}_E \right) \\
&\quad = \frac{1}{\sqrt{2\,(n\!+\!1)(1\!-\!\lambda)}}\sum_{\ell=0}^{n+1} \sqrt{\!\binom{n\!+\!1}{\ell}\!} \left(1\!-\!\lambda\right)^{\frac{\ell}{2}}\lambda^{\frac{n-\ell}{2}} \\
&\quad \hspace{24ex}  \times \left( - \left((n\!+\!1)(1\!-\!\lambda)-\ell \right) \ket{0}_{\!A}\! \ket{\ell}_{\!B}\! \ket{n\!+\!1\!-\!\ell}_{\!E} + \sqrt{(1\!-\!\lambda)(n\!-\!\ell\!+\!1)}\, \ket{1}_{\!A}\!\ket{\ell}_{\!B}\!\ket{n\!-\!\ell}_{\!E} \right) \\
&\quad = \frac{1}{\sqrt{2\,(n\!+\!1)(1\!-\!\lambda)}} \sum_{\ell=0}^{n+1} \sqrt{\!\binom{n\!+\!1}{\ell}\!} \left(1\!-\!\lambda\right)^{\frac{\ell}{2}}\lambda^{\frac{n-\ell}{2}} \left( - \left((n\!+\!1)(1\!-\!\lambda)-\ell\right) \ket{0}_{\!A} \!\ket{\ell}_{\!B} + \sqrt{\lambda\ell}\,\ket{1}_{\!A}\!\ket{\ell\!-\!1}_{\!B}\right) \ket{n\!+\!1\!-\!\ell}_E , 
\end{align*}
with the convention that $\ket{-1}\equiv 0$. 
%&\quad= -\frac{\lambda^{\frac{n}{2}}}{\sqrt{2}}\, \sqrt{(n+1)(1-\lambda)}\, \ket{0}_A\ket{0}_B\ket{n+1}_E \\
%&\quad\quad + \frac{1}{\sqrt{2}}\sum_{\ell=1}^{n+1} \sqrt{\binom{n}{\ell\!-\!1}} (1-\lambda)^{\frac{\ell-1}{2}} \lambda^{\frac{n-\ell+1}{2}}\left( \frac{\ell\!-\!(n+1)(1-\lambda)}{\sqrt{\lambda\ell}} \ket{0}_{\!A} \ket{\ell}_{\!B} + \ket{1}_{\!A}\ket{\ell-1}_{\!B}\right) \ket{n+1-\ell}_E .
Introducing the normalised vectors
\begin{align}
\ket{\zeta_\ell(n,\lambda)} &\coloneqq \frac{1}{\sqrt{(1\!-\!\lambda)(n\!-\!\ell\!+\!1) + \left((n\!+\!1)(1\!-\!\lambda)-\ell\right)^2}} \left( - \left( (n\!+\!1)(1\!-\!\lambda)-\ell\right) \ket{0}\! \ket{n\!+\!1\!-\!\ell} + \sqrt{(1\!-\!\lambda)(n\!-\!\ell\!+\!1)}\, \ket{1}\! \ket{n\!-\!\ell} \right) \label{zeta ell} \\
\ket{\eta_\ell(n,\lambda)} &\coloneqq \frac{1}{\sqrt{\lambda\ell + \left((n\!+\!1)(1\!-\!\lambda) - \ell\right)^2}} \left( - \left((n\!+\!1)(1\!-\!\lambda) - \ell\right) \ket{0}\! \ket{\ell} + \sqrt{\lambda \ell}\, \ket{1}\! \ket{\ell\!-\!1}\right) , \label{eta ell} 
\end{align}
for $\ell=0,\ldots, n+1$, we finally arrive at
\begin{align}
&U_\lambda^{BE} \ket{\Psi}_{\!AB}\!\ket{n}_{\!E} \nonumber \\
&\quad = \frac{1}{\sqrt{2\,(n\!+\!1)(1\!-\!\lambda)}} \sum_{\ell=0}^{n+1} \sqrt{\!\binom{n\!+\!1}{\ell}\!} \left(1\!-\!\lambda\right)^{\frac{\ell}{2}}\lambda^{\frac{n-\ell}{2}} \sqrt{(1\!-\!\lambda)(n\!-\!\ell\!+\! 1) + \left( (n\!+\!1)(1\!-\!\lambda) - \ell \right)^2}\, \ket{\zeta_\ell(n,\lambda)}_{\!AE} \ket{\ell}_{\!B} . \label{Phi n output ABE alt} \\ 
&\quad = \frac{1}{\sqrt{2\,(n\!+\!1)(1\!-\!\lambda)}} \sum_{\ell=0}^{n+1} \sqrt{\!\binom{n\!+\!1}{\ell}\!} \left(1-\lambda\right)^{\frac{\ell}{2}}\lambda^{\frac{n-\ell}{2}} \sqrt{\lambda\ell + \left( (n\!+\!1)(1\!-\!\lambda)-\ell \right)^2}\, \ket{\eta_\ell(n,\lambda)}_{\!AB} \ket{n\!+\!1\!-\!\ell}_{\!E} . \label{Phi n output ABE} 
\end{align}
Tracing away the subsystem $E$ from~\eqref{Phi n output ABE} yields the output state of the channel as
\bb
\begin{aligned}
\omega_{AB}(n,\lambda) \coloneqq&\ \left( I^A\! \otimes \Phi_{\lambda,\, \sigma(n)}^B\right) \left( \Psi_{AB}\right) \\
=&\ \frac{1}{2(n\!+\!1)(1\!-\!\lambda)} \sum_{\ell=0}^{n+1} \binom{n\!+\!1}{\ell} \left(1\!-\!\lambda\right)^{\ell}\lambda^{n-\ell} \left( \lambda\ell + \left( (n\!+\!1)(1\!-\!\lambda) - \ell\right)^2\right) \ketbra{\eta_\ell(n,\lambda)}_{AB}\, .
\end{aligned}
\label{Phi n output AB}
\ee
%\bb \begin{aligned} \left( I^A \otimes \Phi_{\lambda,\, \sigma(n)}^B\right) \left( \Psi_{AB}\right) &= \frac{\lambda^{n}}{2}\, (n+1)(1-\lambda)\, \ketbra{0}_A \otimes \ketbra{0}_B \\ &\quad + \frac{1}{2}\sum_{\ell=1}^{n+1} \binom{n}{\ell-1} (1-\lambda)^{\ell-1} \lambda^{n-\ell+1} \left( 1+ \frac{1}{\lambda\ell} \left(\ell-(n+1)(1-\lambda)\right)^2\right) \ketbra{\eta_\ell}_{AB} \, . \end{aligned} \ee
Note that the total photon number of the state $\ket{\eta_\ell}$ is exactly $\ell$, for all $\ell=0,\ldots, n+1$:
\bb
(a^\dag a + b^\dag b) \ket{\eta_\ell}_{AB} = \ell \ket{\eta_\ell}_{AB}\, .
\ee
Hence, the vectors $\ket{\eta_\ell}$ are all orthogonal to each other. This allows us to immediately deduce the spectrum of $\omega_{AB}(n,\lambda)$. We obtain that
\bb
\mathrm{sp}\left( \omega_{AB}(n,\lambda) \right) = \left\{q_0(n,\lambda),\ldots, q_{n+1}(n,\lambda)\right\} ,
\ee
where the probability distribution $q(n,\lambda)$ is given by~\eqref{q ell}.

To derive an expression for $\omega_B(n,\lambda) = \Phi_{\lambda,\, \sigma(n)}^B \left( \Psi_{B}\right)$ we could trace away $A$ from~\eqref{Phi n output AB}. However, it is slightly more convenient to read off the result directly from~\eqref{Phi n output ABE alt}. We obtain that
\bb
\omega_B(n,\lambda) = \Phi_{\lambda,\, \sigma(n)}^B \left( \Psi_{B}\right) = \frac{1}{2\,(n\!+\!1)(1\!-\!\lambda)} \sum_{\ell=0}^{n+1} \binom{n\!+\!1}{\ell} \left(1\!-\!\lambda\right)^{\ell}\lambda^{n-\ell} \left((1\!-\!\lambda)(n\!-\!\ell\!+\! 1) + \left( (n\!+\!1)(1\!-\!\lambda) - \ell \right)^2\right) \ketbra{\ell}_{\!B}
\ee
The above decomposition allows us to write down the spectrum of the reduced output state on the $B$ system immediately. We obtain that
\bb
\mathrm{sp}\left( \omega_B(n,\lambda) \right) = \left\{p_0(n,\lambda),\ldots, p_{n+1}(n,\lambda)\right\} ,
\ee
where the probability distribution $p(n,\lambda)$ is given by~\eqref{p ell}.

Since the reduced input state $\Psi_B$ on the $B$ system has mean photon number $1/2$, the (energy-constrained) LSD theorem~\eqref{energy-constrained Q 1}--\eqref{energy-constrained Q 2} yields the estimate in~\eqref{Phi n Q estimate}, thus concluding the proof.
\end{proof}

\subsection{Sorting $\boldsymbol{q(n,\lambda)}$} \label{sorting q subsec}

In the following, for a given probability distribution $r=\left(r_0,\ldots, r_N\right)$, we denote with $r^\uparrow=\left( r^\uparrow_0,\ldots, r^\uparrow_N\right)$ the distribution obtained by sorting it in ascending order, so that e.g.\ $r^\uparrow_0 = \min_{\ell=0,\ldots, N} r_\ell$. Our first result tells us that for a wide range of values of $\lambda$ the distribution $q(n,\lambda)$ is actually already sorted. It is useful to define the two functions
\begin{align}
\lambda_+(n) &\coloneqq \frac{3}{n+2} \left( 1 -\sqrt{\frac{n-1}{3(n+1)}}\right) , \label{lambda+}\\
\lambda_-(n) &\coloneqq \frac{2}{n+2} \left( 1 -\sqrt{\frac{n}{2(n+1)}}\right) . \label{lambda-}
\end{align}
We are now ready to state and prove our first result.

\begin{prop} \label{q sorted prop}
For all $n\geq 2$, 
\bb
q^\uparrow(n,\lambda) = q(n,\lambda)\qquad \forall\quad \frac{1}{n+1} \leq \lambda \leq \lambda_+(n) \, .
\ee
\end{prop}

\begin{proof}
For $\ell=0,\ldots, n$, leveraging the fact that
\bb
\frac{\binom{n+1}{\ell+1}}{\binom{n+1}{\ell}} = \frac{n-\ell+1}{\ell+1}
\label{ratio binomial}
\ee
the formula~\eqref{q ell} yields
\begin{align*}
&\lambda (\ell+1)\left( \lambda\ell + \left((n+1)(1-\lambda)-\ell\right)^2\right) \left( \frac{q_{\ell+1}(n,\lambda)}{q_\ell(n,\lambda)} - 1 \right) \\
&\qquad = (n-\ell+1)(1-\lambda) \left( \lambda (\ell+1) + \left( (n+1)(1-\lambda) -  \ell - 1\right)^2 \right) - \lambda (\ell+1)\left( \lambda\ell + \left((n+1)(1-\lambda) - \ell\right)^2\right) \\
&\qquad = \lambda (\ell+1) \left( (n-\ell+1)(1-\lambda) - \lambda\ell - \left((n+1)(1-\lambda)-\ell\right)^2 \right) + (n-\ell+1)(1-\lambda) \left( (n+1)(1-\lambda) - \ell - 1\right)^2 \\
&\qquad = - \lambda (\ell+1) \left( (n+1)(1-\lambda) - \ell -1 \right) \left( (n+1)(1-\lambda) - \ell \right) + (n-\ell+1)(1-\lambda) \left( (n+1)(1-\lambda) - \ell-1 \right)^2 \\
&\qquad = \left( (n+1)(1-\lambda) - \ell - 1 \right)\, \Big( - \lambda (\ell+1) \left( (n+1)(1-\lambda) - \ell \right) + (n-\ell+1)(1-\lambda) \left( (n+1)(1-\lambda) - \ell-1 \right) \Big) \\
&\qquad = \left( (n+1)(1-\lambda) - \ell - 1 \right) \left( \ell^2 - 2\left((n+1)(1-\lambda) - \frac12\right) \ell + (n+1)(1-\lambda) \left( n - (n+2)\lambda\right) \right) .
\end{align*}
Setting
\bb
f_{n,\lambda}(\ell) \coloneqq \ell^2 - 2\left((n+1)(1-\lambda) - \frac12\right) \ell + (n+1)(1-\lambda) \left( n - (n+2)\lambda\right) ,
\label{q sorted f}
\ee
we arrive at the identity
\bb
\lambda (\ell+1)\left( \lambda\ell + \left((n+1)(1-\lambda)-\ell\right)^2\right) \left( \frac{q_{\ell+1}(n,\lambda)}{q_\ell(n,\lambda)} - 1 \right) = \left( (n+1)(1-\lambda) - \ell - 1 \right)\, f_{n,\lambda}(\ell)\, .
\label{q sorted main identity}
\ee
Now, the function $f_{n,\lambda}(\ell)$ is a second-degree polynomial in the variable $\ell$. By finding its roots we can determine its sign on the whole real line. We see that
\bbb
\begin{array}{ll} f_{n,\lambda}(\ell) \leq 0 &\quad \text{if $\ell_-(n,\lambda)\leq \ell \leq \ell_+(n,\lambda)$,} \\[1ex]
f_{n,\lambda}(\ell) \geq 0 &\quad \text{otherwise,}
\end{array}
\eee
where
\bbb
\ell_{\pm}(n,\lambda) \coloneqq n+\frac12 - (n+1)(1-\lambda) \pm \sqrt{\frac14 + (n+1)\lambda(1-\lambda)}\, .
\eee
One can show that
\bbb
\ell_-(n,\lambda)\geq n-2 \qquad \forall\quad 0\leq \lambda \leq \min\left\{ \frac{5}{2(n+1)},\, \lambda_+(n) \right\} = \lambda_+(n)\, .
\eee
Moreover,
\bbb
\ell_-(n,\lambda)\leq n-1 \qquad \forall\quad \lambda \geq \lambda_-(n)\, .
\eee
Putting all together, we find that
\bbb
n-2\leq \ell_-(n,\lambda) \leq n-1\qquad \forall\quad \lambda_-(n)\leq \lambda\leq\lambda_+(n)\, .
\eee
It is also easy to verify that
\bbb
\ell_+(n,\lambda)\geq n\qquad \forall\quad 0\leq \lambda\leq \frac{2}{n+2}\, .
\eee
Since $\frac{2}{n+2}\geq \lambda_+(n)$ for all $n\geq 2$, we deduce that
\bbb
\ell_+(n,\lambda)\geq n\qquad \forall\quad 0\leq \lambda\leq \lambda_+(n)\, .
\eee
Going back to the function $f_{n,\lambda}(\ell)$, the above discussion implies that
\bb
\begin{array}{ll} f_{n,\lambda}(\ell) \leq 0 &\quad \text{if $\ell=n-1,n$,} \\[1ex]
f_{n,\lambda}(\ell) \geq 0 &\quad \text{if $\ell=0,\ldots,n-2$,}
\end{array} \qquad\quad \forall\quad \lambda_-(n)\leq \lambda\leq \lambda_+(n)\, .
\label{q sorted study f}
\ee
Also, it is not difficult to verify that
\bbb
n-2\leq (n+1)\left(1-\lambda\right) - 1\leq n-1\qquad \forall\quad \frac{1}{n+1}\leq \lambda\leq \frac{2}{n+1}\, ;
\eee
we infer that
\bb
\begin{array}{ll} (n+1)\left(1-\lambda\right) - \ell - 1 \leq 0 &\quad \text{if $\ell=n-1,n$,} \\[1ex]
(n+1)\left(1-\lambda\right) - \ell - 1 \geq 0 &\quad \text{if $\ell=0,\ldots,n-2$,}
\end{array} \qquad\quad \forall\quad \frac{1}{n+1}\leq \lambda\leq \frac{2}{n+1}\, .
\label{q sorted study other factor}
\ee
Using the fact that $\frac{1}{n+1}\geq \lambda_-(n)$ and $\lambda_+(n)\leq \frac{2}{n+1}$ for all $n$, and combining~\eqref{q sorted main identity} on the one hand with~\eqref{q sorted study f}--\eqref{q sorted study other factor} on the other, we finally see that
\bbb
q_{\ell+1}(n,\lambda)\geq q_\ell(n,\lambda)\qquad \forall\quad \frac{1}{n+1}\leq \lambda\leq \lambda_+(n)\, ,
\eee
which proves the claim.
\end{proof}

\subsection{Sorting $\boldsymbol{p(n,\lambda)}$} \label{sorting p subsec}

As it turns out, for an analogous range of values of $\lambda$ the probability distribution $p(n,\lambda)$, unlike $q(n,\lambda)$, is \emph{not} automatically sorted in ascending order. The next lemma represents a first step in the direction of ascertaining how $p(n,\lambda)$ can be sorted.

\begin{lemma} \label{p sorted lemma1}
For all $n\geq 2$,
\bb
p_0(n,\lambda) \leq p_1(n,\lambda)\leq \ldots \leq p_{n-1}(n,\lambda) \geq p_n(n,\lambda)\leq p_{n+1}(n,\lambda)\qquad \forall\quad \frac{1}{n+1} \leq \lambda \leq \lambda_+(n)\, .
\ee
\end{lemma}

\begin{proof}
For all $\ell = 0,\ldots, n$, employing~\eqref{p ell} and~\eqref{ratio binomial} we compute
\begin{align*}
&\lambda (\ell+1) \left((1-\lambda)(n-\ell+1) + \left( (n+1)(1-\lambda) - \ell\right)^2 \right) \left( \frac{p_{\ell+1}(n,\lambda)}{p_\ell(n,\lambda)} - 1\right) \\
&\qquad = (n-\ell+1)(1-\lambda) \left( (1-\lambda)(n-\ell) + \left((n+1)(1-\lambda) -\ell -1\right)^2 \right) \\
&\qquad \quad - \lambda (\ell+1) \left( (1-\lambda)(n-\ell+1) + \left((n+1)(1-\lambda) -\ell\right)^2 \right) \\
&\qquad = (n-\ell+1)(1-\lambda) \left( (1-\lambda)(n-\ell) + \left((n+1)(1-\lambda) -\ell -1\right)^2 - \lambda (\ell+1) \right) - \lambda (\ell+1) \left((n+1)(1-\lambda) -\ell\right)^2 \\
&\qquad = (n-\ell+1)(1-\lambda) \left( (n+1)(1-\lambda) - \ell\right) \left( (n+1)(1-\lambda) -\ell -1 \right) - \lambda (\ell+1) \left((n+1)(1-\lambda) -\ell\right)^2 \\
&\qquad = \left((n+1)(1-\lambda) - \ell\right) \, \Big( (n-\ell+1)(1-\lambda) \left( (n+1)(1-\lambda) -\ell -1 \right) - \lambda (\ell+1) \left((n+1)(1-\lambda) -\ell\right) \Big) \\
&\qquad = \left((n+1)(1-\lambda) - \ell\right) \left( \ell^2 - 2\left((n+1)(1-\lambda) - \frac12\right) \ell + (n+1)(1-\lambda) \left( n - (n+2)\lambda\right) \right) .
\end{align*}
Thus,
\bb
\lambda (\ell+1) \left((1-\lambda)(n-\ell+1) + \left( (n+1)(1-\lambda) - \ell\right)^2 \right) \left( \frac{p_{\ell+1}(n,\lambda)}{p_\ell(n,\lambda)} - 1\right) = \left((n+1)(1-\lambda) - \ell\right) \, f_{n,\lambda}(\ell)\, ,
\label{p sorted main identity}
\ee
where $f_{n,\lambda}(\ell)$, defined by~\eqref{q sorted f}, is -- luckily enough -- the same function that we already encountered in the proof of Proposition~\ref{q sorted prop}, which makes~\eqref{q sorted study f} available. Since
\bbb
n-1\leq (n+1)\left(1-\lambda\right) \leq n\qquad \forall\quad \frac{1}{n+1}\leq \lambda\leq \frac{2}{n+1}\, ,
\eee
we obtain that
\bb
\begin{array}{ll} (n+1)\left(1-\lambda\right) - \ell \leq 0 &\quad \text{if $\ell=n$,} \\[1ex]
(n+1)\left(1-\lambda\right) - \ell - 1 \geq 0 &\quad \text{if $\ell=0,\ldots,n-1$,}
\end{array} \qquad\quad \forall\quad \frac{1}{n+1}\leq \lambda\leq \frac{2}{n+1}\, .
\label{p sorted study other factor}
\ee
Combining~\eqref{p sorted main identity} with~\eqref{q sorted study f} and~\eqref{p sorted study other factor} shows that for all $\frac{1}{n+1}\leq \lambda\leq \lambda_+(n)$ the inequalities $p_{\ell}(n,\lambda)\leq p_{\ell+1}(n,\lambda)$ hold true for $\ell=0,\ldots, n-2$ or $\ell=n$, while for $\ell=n-1$ we have the opposite relation $p_{n-1}(n,\lambda)\geq p_{n}(n,\lambda)$. This completes the proof.
\end{proof}

\begin{lemma} \label{p sorted lemma2}
For all $n\geq 3$,
\bb
p_{n-3}(n,\lambda) \leq p_n(n,\lambda)\qquad \forall\quad \frac{1}{n+1}\leq \lambda \leq \widetilde{\lambda}_+(n)\, ,
\ee
where
\bb
\widetilde{\lambda}_+(n) \coloneqq \frac{3^{1/3}}{2^{1/3} n + 3^{1/3} - 2^{1/3}}\, .
\label{lambda+ tilde}
\ee
\end{lemma}

\begin{proof}
Using the explicit formulae~\eqref{p ell} and~\eqref{q ell}, we compute
\begin{align*}
\frac{p_{n-3}(n,\lambda)}{p_n(n,\lambda)} &= \frac{1}{24}\, n(n-1)(n-2) \left( \frac{\lambda}{1-\lambda}\right)^3 \frac{4(1-\lambda)+\left(4-(n+1)\lambda \right)^2}{1-\lambda+\left(1-(n+1)\lambda \right)^2} \\
&= \left( \frac{2}{3}\, n(n-1)(n-2) \left( \frac{\lambda}{1-\lambda}\right)^3 \right) \left( \frac{1}{16}\, \frac{4(1-\lambda)+\left(4-(n+1)\lambda \right)^2}{1-\lambda+\left(1-(n+1)\lambda \right)^2} \right)
\end{align*}
We now evaluate separately the above two factors, and show that they are both upper bounded by $1$. The first one can be estimated by resorting to the elementary inequality $n(n-2)\leq (n-1)^2$; one obtains that
\bbb
\frac{2}{3}\, n(n-1)(n-2) \left( \frac{\lambda}{1-\lambda}\right)^3 \leq \frac{2}{3}\, (n-1)^3 \left( \frac{\lambda}{1-\lambda}\right)^3\leq 1\qquad \forall\quad 0\leq \lambda\leq \widetilde{\lambda}_+(n)\, ,
\eee
where the last inequality can be easily proved by taking the cubic root of both sides. Upon simple algebraic manipulations, the inequality 
\bbb
\frac{1}{16}\, \frac{4(1-\lambda)+\left(4-(n+1)\lambda \right)^2}{1-\lambda+\left(1-(n+1)\lambda \right)^2} \leq 1\, ,
\eee
which is to be proved, becomes
\bbb
\frac{5}{4} (1+n)^2 \lambda ^2 - (3+2 n) \lambda - 1 \geq 0\, .
\eee
The discriminant of the second-degree polynomial on the left-hand side is $-n^2+2n+4$. This is negative for all $n\geq 4$, and hence in this case the above inequality is satisfied for all $0\leq \lambda\leq 1$ and a fortiori in the prescribed range. If $n=3$, an explicit calculation shows that the inequality holds true for $\lambda\leq \frac15$ or $\lambda\geq \frac14$, i.e.\ in particular for all $\lambda\geq \frac{1}{n+1}=\frac14$. This completes the proof.
\end{proof}

\begin{lemma} \label{max p ell lemma}
For all $n\geq 2$,
\bbb
p_{n-1}(n,\lambda) = \max_{\ell=0,\ldots, n+1} p_\ell(n,\lambda)\qquad \forall\quad \frac{1}{n+1}\leq \lambda\leq \lambda_+(n)\, .
\eee
\end{lemma}

\begin{proof}
Since Lemma~\ref{p sorted lemma1} holds in the prescribed interval in $\lambda$, we need only to prove that $p_{n-1}(n,\lambda)\geq p_{n+1}(n,\lambda)$. Indeed, one verifies that
\bbb
(n+1)^2(1-\lambda)^2\left( \frac{p_{n-1}(n,\lambda)}{p_{n+1}(n,\lambda)} - 1 \right) = \frac{1}{2} (n-1) (n+1) (n+2) \lambda^2 -\left(2n^2+n-2\right) \lambda +2n -1\geq 0\, ,
\eee
where the last inequality holds because the above second-degree polynomial in $\lambda$ has discriminant $n(-2n^2+n+2)<0$ as soon as $n\geq 2$.
\end{proof}

\begin{prop} \label{p sorted prop}
For all $n\geq 2$ and all $\frac{1}{n+1}\leq \lambda\leq \frac1n$, 
\bb
\begin{array}{rl}
\text{either} &\qquad p^\uparrow(n,\lambda) = \Big( p_0(n,\lambda),\, p_1(n,\lambda), \ldots,\, p_{n-3}(n,\lambda),\, p_n(n,\lambda),\, p_{n+1}(n,\lambda),\, p_{n-2}(n,\lambda),\, p_{n-1}(n,\lambda) \Big) \, , \\[1.5ex]
\text{or} &\qquad p^\uparrow(n,\lambda) = \Big( p_0(n,\lambda),\, p_1(n,\lambda), \ldots,\, p_{n-3}(n,\lambda),\, p_n(n,\lambda),\, p_{n-2}(n,\lambda),\, p_{n+1}(n,\lambda),\, p_{n-1}(n,\lambda) \Big) \, , \\[1.5ex]
\text{or} &\qquad p^\uparrow(n,\lambda) = \Big( p_0(n,\lambda),\, p_1(n,\lambda), \ldots,\, p_{n-3}(n,\lambda),\, p_{n-2}(n,\lambda),\, p_n(n,\lambda),\, p_{n+1}(n,\lambda),\, p_{n-1}(n,\lambda) \Big) \, .
\end{array}
\ee
When $n=2$, it is understood that only the last $4$ entries are to be taken into account.
\end{prop}

\begin{proof}
It suffices to combine Lemmata~\ref{p sorted lemma1},~\ref{p sorted lemma2} and~\ref{max p ell lemma}. Note that $\frac1n \leq \min\left\{\lambda_+(n),\, \widetilde{\lambda}_+(n)\right\}$ for all $n\geq 2$.
\end{proof}

%The above Lemma~\ref{p sorted lemma3} and Proposition~\ref{p sorted prop} cannot be improved upon for $n=2,3$ while keeping the lower end of the interval in $\lambda$ equal to $1/(n+1)$. In fact, numerical investigations reveal that Lemma~\ref{p sorted lemma3}, and Proposition~\ref{p sorted prop} with it, break down in a right neighbourhood of $1/(n+1)$ for $n=2,3$. However, it is nevertheless the case that $p_{n-1}(n,\lambda) = \max_{\ell} p_\ell(n,\lambda)$ in the same range of values of $\lambda$. We formalise this fact as a lemma for future convenience.

\subsection{Majorisation} \label{majorisation subsec}

Let $r=(r_0,\ldots, r_N)$ and $s=(s_0,\ldots, s_N)$ be two probability distributions. We remind the reader that $r$ is said to be \emph{majorised} by $s$, and we write $r\prec s$, if
\bb
\sum_{\ell=0}^k r^\uparrow_\ell \geq \sum_{\ell=0}^k s^\uparrow_\ell\qquad \forall\ k=0,\ldots, N\, .
\label{maj}
\ee
Of course, the above inequality becomes an equality for $k=N$, since the elements of both distributions add up to $1$. For a complete introduction to the theory of majorisation, we refer the reader to the excellent monograph by Marshall and Olkin~\cite{MARSHALL-OLKIN}.

The goal of this subsection is to show that the two probability distributions $p(n,\lambda)$ and $q(n,\lambda)$ obey precisely a majorisation relation $p(n,\lambda)\prec q(n,\lambda)$. Our first step in this direction is a simple lemma.

\begin{lemma} \label{p ell < q ell}
For all $n\geq 2$,
\bb
q_\ell(n,\lambda)\leq p_\ell(n,\lambda) \qquad \forall\quad \ell=0,\ldots, n-1\, , \qquad \forall\quad 0\leq \lambda\leq \frac{2}{n+1}\, .
\ee
\end{lemma}

\begin{proof}
Using the expressions~\eqref{p ell} and~\eqref{q ell}, one verifies that
\begin{align*}
&\left( \lambda\ell + \left( (n+1)(1-\lambda) - \ell \right)^2\right) \left( \frac{p_\ell(n,\lambda)}{q_\ell(n,\lambda)} -1\right) \\
&\qquad = (1-\lambda)(n-\ell+ 1) + \left( (n+1)(1-\lambda) - \ell \right)^2 - \lambda\ell - \left( (n+1)(1-\lambda) - \ell \right)^2 \\
&\qquad = n-\ell+1-(n+1)\lambda \\
&\qquad \geq 0\, ,
\end{align*}
where the last inequality holds provided that $\lambda\leq \frac{2}{n+1}$ and $\ell\leq n-1$.
\end{proof}

\begin{lemma} \label{p vs q sorted n+1 lemma}
For all $n\geq 3$,
\bb
q_{n+1}(n,\lambda) - p_{n-1}(n,\lambda) \geq \frac{(n+1)(n-2)}{4n(n-1)} \left(1-\frac{1}{n} \right)^{n} \qquad \forall\quad 0\leq \lambda\leq \frac1n\, .
\label{p vs q sorted n+1}
\ee
When $n=2$, we have instead that
\bb
q_{3}(2,\lambda) - p_{1}(2,\lambda) \geq \frac{\epsilon}{4} \qquad \forall\quad 0\leq \lambda\leq \frac12-\epsilon
\label{p vs q sorted n+1 special case n=2}
\ee
for any fixed $\epsilon>0$.
\end{lemma}

\begin{proof}
For all $n\geq 2$, one verifies that
\begin{align*}
&\frac{\partial}{\partial \lambda} \left( q_{n+1}(n,\lambda) - p_{n-1}(n,\lambda) \right) \\
&\qquad = \frac{1}{4} (1-\lambda )^{n-3} \left(2-6 n+2 \left(6n^2+n-3\right) \lambda - (n+1)^2 (7n-6) \lambda^2 + (n+1)^2 \left(n^2+n-2\right) \lambda^3 \right) \\
&\qquad \eqqcolon \frac{1}{4} (1-\lambda )^{n-3} g_n(\lambda)\, .
\end{align*}
Now, since
\begin{align*}
\frac{d^2 g_n(\lambda)}{d\lambda} &= 2 (n+1)^2 \left(6 - 7 n+ 3 \left(n^2+n-2\right) \lambda \right) \\
&\leq 2 (n+1)^2 \left(6 - 7 n+ 3 \left(n^2+n-2\right) \frac1n \right) \\
&= -\frac2n (n+1)^2 \left( 4n^2 - 9n+6 \right) \leq 0
\end{align*}
for $0\leq \lambda\leq \frac1n$, the first derivative $\frac{d g_n(\lambda)}{d\lambda}$ of $g_n(\lambda)$ is a decreasing function of $\lambda$ in the same interval $\left[ 0, \frac1n\right]$ . Hence,
\bbb
\min_{\frac{1}{n+1}\leq \lambda\leq \frac1n} \frac{dg_n(\lambda)}{d\lambda} = \frac{dg_n(\lambda)}{d\lambda} \Big|_{\lambda=\frac1n} = \frac{n+1}{n^2}\left( 6 + n(n-1)(n-3) \right) \geq 0\, .
\eee
In turn, this implies that
\bbb
\max_{\frac{1}{n+1}\leq \lambda\leq\frac1n}g_n(\lambda) = g_n\left( \frac1n\right) = -\frac{(n-1)^2(n+2)}{n^3} \leq 0\, .
\eee
Thus, $q_{n+1}(n,\lambda) - p_{n-1}(n,\lambda)$ is decreasing in $\lambda$. Finally, we conclude from this that
\bbb
\min_{\frac{1}{n+1}\leq \lambda\leq\frac1n} \left\{q_{n+1}(n,\lambda) - p_{n-1}(n,\lambda)\right\} = q_{n+1}\left(n,\frac1n\right) - p_{n-1}\left(n,\frac1n\right) = \frac{(n+1)(n-2)}{4n(n-1)} \left(1-\frac{1}{n} \right)^{n} \, .
\eee
When $n=2$, we have instead that
\bbb
\min_{\frac{1}{3}\leq \lambda\leq\frac12-\epsilon} \left\{q_{3}(2,\lambda) - p_{1}(2,\lambda)\right\} = q_{3}\left(2,\,\frac12-\epsilon\right) - p_{1}\left(2,\,\frac12-\epsilon\right) = \frac{\epsilon}{4} + 3\epsilon^3\geq \frac{\epsilon}{4}\, .
\eee
This concludes the proof.
\end{proof}

\begin{prop} \label{majorisation prop}
For all $n\geq 2$,
\bb
p(n,\lambda)\prec q(n,\lambda)\qquad \forall\quad \frac{1}{n+1}\leq \lambda\leq \frac1n\, .
\ee
\end{prop}

\begin{proof}
According to~\eqref{maj}, we need to verify that
\bb
\sum_{\ell=0}^k p^\uparrow_\ell(n,\lambda) \geq \sum_{\ell=0}^k q^\uparrow_\ell(n,\lambda)\qquad \forall\quad k=0,\ldots, n\, ,\qquad \forall\quad \frac{1}{n+1}\leq \lambda\leq \frac1n\, ,
\label{majorisation proof eq1}
\ee
where we used the fact that the inequality corresponding to $k=n+1$ is in fact an equality, by normalisation. Using Proposition~\ref{p sorted prop} and Lemma~\ref{p ell < q ell}, and observing that $\frac1n\leq \frac{2}{n+1}$ for all $n\geq 1$, we see that (when $n\geq 3$) the cases $k=0,\ldots, n-3$ of~\eqref{majorisation proof eq1} are automatically satisfied. Exploiting again normalisation, we recast the difference of the two sides of~\eqref{majorisation proof eq1} (for arbitrary $k$) as
\bb
\sum_{\ell=0}^{k} p^\uparrow_\ell(n,\lambda) - \sum_{\ell=0}^{k} q^\uparrow_\ell(n,\lambda) = \sum_{\ell=k+1}^{n+1} q^\uparrow_\ell(n,\lambda) - \sum_{\ell=k+1}^{n+1} p^\uparrow_\ell(n,\lambda) = \sum_{\ell=k+1}^{n+1} q_\ell(n,\lambda) - \sum_{\ell=k+1}^{n+1} p^\uparrow_\ell(n,\lambda)\, ,
\label{majorisation proof eq2}
\ee
where the last identity follows from Proposition~\ref{q sorted prop}, once one observes that $\frac1n\leq \lambda_+(n)$ as long as $n\geq 2$. It remains to check the positivity of~\eqref{majorisation proof eq2} for $k=n,n-1,n-2$ and for $\frac{1}{n+1}\leq\lambda\leq \frac1n$. When $k=n-1$ and $k=2$ we have to reckon the (two) distinct possibilities offered by Proposition~\ref{p sorted prop}. This makes a total of $5$ different cases to vet. We break down the proof into the separate analysis of each of these cases.
\begin{itemize}
\item $k=n$. Thanks to Lemma~\ref{max p ell lemma} (or Proposition~\ref{p sorted prop}) and Lemma~\ref{p vs q sorted n+1 lemma},
\bbb
q_{n+1}(n,\lambda) - p^\uparrow_{n+1}(n,\lambda) = q_{n+1}(n,\lambda) - p_{n-1}(n,\lambda) \geq 0\, .
\eee
\item $k=n-1$ and $p^\uparrow_n(n,\lambda) = p_{n-2}(n,\lambda)$. Let us write
\begin{align*}
&q_{n+1}(n,\lambda) + q_{n}(n,\lambda) - p_{n-1}(n,\lambda) - p_{n-2}(n,\lambda) \\
&\quad = \frac{1}{2} (1-\lambda )^{n-3} \left(2 - 3(n+2) \lambda + \left(6+9n+n^2\right) \lambda^2 + \frac12 (n+1) \left( n^2 -10n -4 \right) \lambda^3 - \frac16 n (n+1) (n+2) (n-5) \lambda^4\right) \\
&\quad \eqqcolon \frac{1}{2} (1-\lambda )^{n-3} \, h_n(\lambda)\, .
\end{align*}
Now, since
\bbb
\frac{d^3 h_n(\lambda)}{d\lambda^3} = (n+1) \left( 3 \left( n^2 -10n -4 \right) - 4 n (n+2) (n-5) \lambda \right)
\eee
is a linear function of $\lambda$, we have that
\begin{align*}
\max_{\frac{1}{n+1}\leq \lambda\leq \frac1n} \frac{d^3 h_n(\lambda)}{d\lambda^3} &= \max\left\{ \frac{d^3 h_n(\lambda)}{d\lambda^3}\Big|_{\lambda=\frac{1}{n+1}},\, \frac{d^3 h_n(\lambda)}{d\lambda^3}\Big|_{\lambda=\frac1n} \right\} \\
&= \max\left\{ - n^3 - 15 n^2 - 2n - 12,\, -(n+1)\left( n^2+18n - 28\right)\right\} \\
&\leq 0\, .
\end{align*}
That is to say, the function $\frac{d^2h_n(\lambda)}{d\lambda^2}$ is non-increasing on $\left[ \frac{1}{n+1},\frac1n\right]$. Therefore,
\bbb
\min_{\frac{1}{n+1}\leq \lambda\leq \frac1n} \frac{d^2h_n(\lambda)}{d\lambda^2} = \frac{d^2h_n(\lambda)}{d\lambda^2} \Big|_{\lambda=\frac1n} = 3n^2-5n-4+\frac8n \geq n-4+\frac8n \geq 4\left(\sqrt2 -1\right) \geq 0\, ,
\eee
where we exploited the fact that $n\geq 2$. This shows that the function $\frac{d h_n(\lambda)}{d\lambda}$ is non-decreasing on $\left[ \frac{1}{n+1},\frac1n\right]$. We infer that
\bbb
\max_{\frac{1}{n+1}\leq \lambda\leq \frac1n} \frac{d h_n(\lambda)}{d\lambda} = \frac{d h_n(\lambda)}{d\lambda}\Big|_{\lambda=\frac{1}{n+1}} = -\frac{n+1}{6n^2}\left( n^2+2n+4\right) \leq 0\, .
\eee
Finally, given that $h_n(\lambda)$ has been shown to be non-increasing on $\left[ \frac{1}{n+1},\frac1n\right]$, we have that
\bbb
\min_{\frac{1}{n+1}\leq \lambda\leq \frac1n} h_n(\lambda) = h_n\left( \frac1n \right) = \frac{1}{6n^3} \left(n-1\right)\left(n-2\right) \left(2n-1\right) \geq 0\, ,
\eee
which shows that $q_{n+1}(n,\lambda) + q_{n}(n,\lambda) - p_{n-1}(n,\lambda) - p_{n-2}(n,\lambda)$ for all $\frac{1}{n+1}\leq \lambda\leq \frac1n$ and concludes the analysis of this case.

\item $k=n-1$ and $p^\uparrow_n(n,\lambda) = p_{n+1}(n,\lambda)$. We compute
\begin{align*}
&q_{n+1}(n,\lambda) + q_{n}(n,\lambda) - p_{n-1}(n,\lambda) - p_{n+1}(n,\lambda) \\
&\qquad = (1-\lambda )^{n-2} \left( 1 - \frac12 \left(4n+5 \right) \lambda +\frac12 \left( 3n^2+6n+4 \right) \lambda^2 - \frac14 (n+1)^2(n+2) \lambda^3 \right) \\
&\qquad \eqqcolon (1-\lambda )^{n-2} \, s_n(\lambda)\, .
\end{align*}

Let us first deal with the case $n=2$; note that $s_2(\lambda) = 1-\frac{13}{2}\, \lambda + 14 \lambda^2 - 9 \lambda^3$. Now, $\frac{d s_2(\lambda)}{d\lambda} = -\frac{13}{2} + 28\lambda - 27\lambda^2\geq 0$ for $0.351\approx \frac{28-\sqrt{82}}{54}\leq \lambda\leq \frac{28+\sqrt{82}}{54} \approx 0.686$, and $\frac{ds_2(\lambda)}{d\lambda}\leq 0$ outside of that interval. Hence,
\bbb
\min_{\frac13 \leq \lambda\leq \frac12} s_2(\lambda) = s_2\left( \frac{28-\sqrt{82}}{54}\right) \approx 0.054 \geq 0\, .
\eee

We now consider the case where $n\geq 3$. Since
\bbb
\frac{d^2s_n(\lambda)}{d\lambda^2} = 3n^2+6n+4 - \frac32 (n+1)^2 (n+2) \lambda
\eee
is decreasing in $\lambda$, we obtain that
\bbb
\min_{\frac{1}{n+1}\leq \lambda\leq \frac1n} \frac{d^2s_n(\lambda)}{d\lambda^2} = \frac{d^2s_n(\lambda)}{d\lambda^2} \Big|_{\lambda=1/n} = \frac32\, n^2 - \frac3n - \frac72 \geq \frac{27}{2} - 1 - \frac72 = 9 \geq 0\, ,
\eee
where we used the fact that $n\geq 3$. This proves that $\frac{ds_n(\lambda)}{d\lambda}$ is non-decreasing on $\left[ \frac{1}{n+1},\frac1n\right]$. Hence,
\bbb
\min_{\frac{1}{n+1}\leq \lambda\leq \frac1n} \frac{d s_n(\lambda)}{d\lambda} = \frac{d s_n(\lambda)}{d\lambda} \Big|_{\lambda=\frac{1}{n+1}} = \frac{n(n-3)}{4(n+1)} \geq 0\, ,
\eee
where the last estimate holds because $n\geq 3$. We have just shown that $s_n(\lambda)$ is non-decreasing in the interval $\left[ \frac{1}{n+1},\frac1n\right]$. We infer that
\bbb
\min_{\frac{1}{n+1}\leq \lambda\leq \frac1n} s_n(\lambda) = s_n\left( \frac{1}{n+1} \right) = \frac{n(n-1)}{4(n+1)^2}\geq 0\, ,
\eee
which shows that $q_{n+1}(n,\lambda) + q_{n}(n,\lambda) - p_{n-1}(n,\lambda) - p_{n+1}(n,\lambda)\geq 0$ for all $n\geq 2$ and all $\frac{1}{n+1}\leq \lambda\leq \frac1n$, thus completing the argument for this case.

\item $k=n-2$ and $\left\{ p^\uparrow_{n-1}(n,\lambda),\, p^\uparrow_n(n,\lambda) \right\} = \left\{ p_{n+1}(n,\lambda),\, p_{n-2}(n,\lambda) \right\}$. The relevant quantity is now
\begin{align*}
&q_{n+1}(n,\lambda) + q_{n}(n,\lambda) + q_{n-1}(n,\lambda) - p_{n-1}(n,\lambda) - p_{n+1}(n,\lambda) - p_{n-2}(n,\lambda) \\
&\qquad = (1-\lambda)^{n-3}\left( 1-\left(\frac{7}{2}+n\right)\lambda -\frac{1}{4}\left(n^2-15n-18\right)\lambda ^2+\frac{1}{2}\left(n^3-2n^2-7n-5\right) \lambda ^3-\frac{1}{12} (n+1)^2 (n+2) (n-3) \lambda^4 \right) \\
&\qquad \eqqcolon (1-\lambda)^{n-3}\, t_n(\lambda)\, .
\end{align*} 

To study the polynomial $t_n(\lambda)$, let us treat separately the cases $n=2$ and $n\geq 3$. Note that
\bbb
t_2(\lambda) = (1-\lambda)^2\left( 1 - \frac12 \lambda \left(7 - 6\lambda \right) \right) \geq 0\qquad \forall\quad \frac13\leq \lambda\leq \frac12\, ,
\eee
where the last inequality is a consequence of the fact that the function $\lambda\mapsto 1 - \frac12 \lambda \left(7 - 6\lambda \right)$ is decreasing on $\left(-\infty,\frac{7}{12}\right]\supset \left[\frac13, \frac12\right]$ and vanishes for $\lambda=\frac12$.

We now look at the case where $n\geq 3$. Since $\frac{d^4 t_n(\lambda)}{d\lambda^4} = -2(n+1)^2 (n+2)(n-3)\leq 0$, the function $\frac{d^2 t_n(\lambda)}{d\lambda^2}$ is concave. Hence,
\begin{align*}
\min_{\frac{1}{n+1}\leq \lambda\leq \frac1n} \frac{d^2 t_n(\lambda)}{d\lambda^2} &= \min\left\{ \frac{d^2 t_n(\lambda)}{d\lambda^2} \Big|_{\lambda=\frac{1}{n+1}} ,\, \frac{d^2 t_n(\lambda)}{d\lambda^2} \Big|_{\lambda=\frac1n} \right\} \\
&= \min\left\{ \frac{n}{2(n+1)}\left( 3n^2+2n+5 \right),\, \frac{1}{2n^2} \left(3 n^4+n^3-10 n^2-4 n+12\right) \right\} \\
&\geq 0\, ,
\end{align*}
where in the last step we used the fact that $n\geq 3$. We deduce that $\frac{dt_n(\lambda)}{d\lambda}$ is non-decreasing on $\left[ \frac{1}{n+1},\, \frac1n\right]$, in turn implying that
\bbb
\max_{\frac{1}{n+1}\leq \lambda\leq \frac1n} \frac{d t_n(\lambda)}{d\lambda} = \frac{d t_n(\lambda)}{d\lambda}\Big|_{\lambda=1/n} = -\frac{(n-1)}{n^4} \left(2 n^3 - 2 n^2 - 7 n+12\right) \leq 0\, ,
\eee
where the last inequality holds because $2 n^3 - 2 n^2 - 7 n+12\geq 4 n^2 - 7 n+12\geq 9 \geq 0$ for $n\geq 3$. Since we have just shown that $t_n(\lambda)$ is non-increasing on $\left[ \frac{1}{n+1},\, \frac1n\right]$, we conclude that
\bbb
\min_{\frac{1}{n+1}\leq \lambda\leq \frac1n} t_n(\lambda) = t_n\left( \frac1n \right) = \frac{(n-2) (n-1) (2 n^2 - 2 n+3)}{12 n^4} \geq 0\, ,
\eee
concluding the argument.

\item $k=n-2$ and $\left( p^\uparrow_{n-1}(n,\lambda),\, p^\uparrow_n(n,\lambda) \right) = \left( p_{n}(n,\lambda),\, p_{n+1}(n,\lambda) \right)$. The analysis of this last case is much simpler. It suffices to verify that
\bbb
q_{n+1}(n,\lambda) + q_{n}(n,\lambda) + q_{n-1}(n,\lambda) - p_{n-1}(n,\lambda) - p_{n+1}(n,\lambda) - p_{n}(n,\lambda) = \frac{1}{4} n(n-1) \lambda^2 (1-\lambda )^{n-2} \geq 0\, .
\eee
This completes the proof.

\end{itemize}
\end{proof}

\subsection{Concluding the proof} \label{concluding subsec}

\begin{proof}[Proof of Theorem~\ref{Behemoth thm}]
Let us partition the $(0,1]$ into the three regions
\bb
(0,1] = \left(0,\frac12-\epsilon\right] \cup \left[\frac12-\epsilon,\, \frac12 + \epsilon\right] \cup \left[\frac12 + \epsilon,\, 1\right] ,
\ee
where $\epsilon >0$ is a small constant to be determined later. In the third region, i.e.\ for $\frac12+\epsilon\leq \lambda\leq 1$, the claim follows elementarily from the ansatz $\sigma=\ketbra{0}$, which brings us back to the case of the pure loss channel. Thanks to~\eqref{pure loss Q}, we know that
\bbb
Q\left(\Phi_{\lambda,\, \ket{0}\!\bra{0}},\, \frac12 \right) = Q\left(\mathcal{E}_\lambda,\, \frac12\right) = \max\left\{ g\left(\frac{\lambda}{2}\right) - g\left( \frac{1-\lambda}{2} \right),\, 0 \right\} \geq g\left(\frac12 \left(\frac12 +\epsilon\right) \right) - g\left( \frac12 \left(\frac12-\epsilon\right) \right) > 0
\eee
as long as $\epsilon>0$. In the second region, that is, for $\frac12-\epsilon\leq \lambda\leq \frac12+\epsilon$, one can use Example~\ref{positive capacity at 1/2 ex} and some standard continuity arguments. Namely, consider the state $\sigma=\xi(1/3)$ as defined by~\eqref{xi eta}; specialising~\eqref{positive capacity 1/2} we find that 
\bbb
Q\left( \Phi_{1/2,\, \xi(1/3)},\, \frac12\right) \geq I_{\mathrm{coh}}(A\rangle B)_{\zeta_{AB}(1/2,\,1/3)} \approx 0.07392 > 0 \, ,
\eee
where $\zeta_{AB}(\lambda,\eta)$ is the reduced state on $AB$ corresponding to~\eqref{zeta}. The density matrices $\zeta_{AB}(\lambda,\, 1/3)$ clearly depend continuously on $\lambda$; moreover, they live in a qubit--qutrit system for all values of $\lambda$. Hence, the Alicki--Fannes--Winter inequality~\cite{Alicki-Fannes, tightuniform} implies that $I_{\mathrm{coh}}(A\rangle B)_{\zeta_{AB}(\lambda,\,1/3)}$ is a continuous function of $\lambda$. By choosing $\epsilon>0$ small enough, we can therefore insure that
\bb
Q\left( \Phi_{\lambda,\, \xi(1/3)},\, \frac12\right) \geq I_{\mathrm{coh}}(A\rangle B)_{\zeta_{AB}(\lambda,\,1/3)} \geq c_1 \qquad \forall\quad \frac12-\epsilon\leq \lambda\leq \frac12 +\epsilon\, ,
\ee
where $c_1>0$ is a universal constant.

We are thus left with the first region, corresponding to $0<\lambda\leq \frac12-\epsilon$. We further split it according to
\bbb
\left(0,\, \frac12 - \epsilon\right] = \left[ \frac13,\, \frac12-\epsilon\right] \cup \bigcup_{n=3}^\infty \left[\frac{1}{n+1},\, \frac1n \right] .
\eee
Thanks to Proposition~\ref{Phi n Q estimate prop}, we need only to show that
\bb
\mathcal{I}\left(n,\lambda\right) = H(p(n,\lambda)) - H(q(n,\lambda)) \geq c \qquad \forall\quad \frac{1}{n+1}\leq \lambda\leq \frac1n
\label{Behemoth proof Jn}
\ee
for all $n\geq 3$ and for some universal constant $c_2>0$, and also that
\bb
\mathcal{I}\left(2,\lambda\right) = H(p(2,\lambda)) - H(q(2,\lambda)) \geq c_3 \qquad \forall\quad \frac{1}{3}\leq \lambda\leq \frac12-\epsilon
\label{Behemoth proof J2}
\ee
for some other constant $c_3>0$.

Our main tool here will be a beautiful inequality proved by Ho and Verd\'u~\cite[Theorem~3]{Ho2010}. This states that whenever $r$ and $s$ are two probability distributions such that $r\prec s$, it holds that
\bb
H(s) - H(r) \geq D\left(s^\uparrow \big\|\, r^\uparrow\right)\, ,
\label{Ho-Verdu}
\ee
where $D(u\|v) \coloneqq \sum_\ell u_\ell \log_2 \frac{u_\ell}{v_\ell}$ is the Kullback--Leibler divergence, i.e.\ the relative entropy. Let us first deal with the case $n\geq 3$. We obtain that
\begin{align*}
\mathcal{I}\left(n,\lambda\right) &= H(p(n,\lambda)) - H(q(n,\lambda)) \\
&\textgeq{1} D\big( q^\uparrow (n,\lambda) \big\| p^\uparrow (n,\lambda) \big) \\
&\textgeq{2} \frac{1}{2 \ln 2} \left\|q^\uparrow (n,\lambda) - p^\uparrow (n,\lambda) \right\|_1^2 \\
&\textgeq{3} \frac{2}{\ln 2} \left|q^\uparrow_{n+1} (n,\lambda) - p^\uparrow_{n+1} (n,\lambda) \right|^2 \\
&\texteq{4} \frac{2}{\ln 2} \left|p_{n-1} (n,\lambda) - q_{n+1} (n,\lambda) \right|^2 \\
&\textgeq{5} \frac{2}{\ln 2} \left(\frac{(n+1)(n-2)}{4n(n-1)}\right)^2 \left(1-\frac{1}{n} \right)^{2n} \\
&\textgeq{6} \frac{32}{6561 \ln 2} > 0 \, .
\end{align*}
Here, 1 comes from applying the Ho--Verd\'u inequality~\eqref{Ho-Verdu} to the case of $r=p(n,\lambda)$ and $s=q(n,\lambda)$, which is possible by Proposition~\ref{majorisation prop}. The estimate in 2 is just Pinsker's inequality (see~\cite[p.58]{CSISZAR-KOERNER} and references therein). In 3 we estimated the total variation or $L_1$ distance between $p^\uparrow (n,\lambda)$ and $q^\uparrow (n,\lambda)$ from below as twice their $L_\infty$ distance, namely
\bbb
\left\|q^\uparrow (n,\lambda) - p^\uparrow (n) \right\|_1 \geq 2 \max_{\ell=0,\ldots, n+1} \left|q^\uparrow_\ell (n,\lambda) - p^\uparrow_\ell (n) \right| \geq \left| q^\uparrow_{n+1} (n,\lambda) - p^\uparrow_{n+1} (n) \right| \, .
\eee
Then, in 4 we used Proposition~\ref{q sorted prop} and Lemma~\ref{max p ell lemma}, together with the observation that $\lambda_+(n)\geq \frac1n$ for all $n\geq 2$. The estimate in 5 follows from~\eqref{p vs q sorted n+1}, while in 6 we noted that both $n\mapsto \frac{(n+1)(n-2)}{4n(n-1)}$ and $n\mapsto \left(1-\frac{1}{n} \right)^{n}$ are increasing function of $n$ for $n\geq 3$, and therefore their product can be lower bounded by evaluating it for $n=3$. Thus,~\eqref{Behemoth proof Jn} holds with $c_3=\frac{32}{6561 \ln 2}$.

It remains to deal with the $n=2$ case. We can repeat the same reasoning as above all the way until step 5, where we have to use instead the estimate in~\eqref{p vs q sorted n+1 special case n=2}, thus obtaining
\bbb
\mathcal{I}\left(2,\lambda\right) \geq \frac{\epsilon^2}{8 \ln 2} \eqqcolon c_2 > 0\, ,
\eee
which proves~\eqref{Behemoth proof J2}. Setting $c\coloneqq \min\{c_1,c_2,c_3\}$ completes the argument.
\end{proof}

\begin{rem} \label{value c rem}
The optimal constant in Theorem~\ref{Behemoth thm} can be expressed as a function of the energy constraint $N$ as
\bb
c(N) \coloneqq \inf_{0<\lambda\leq 1} \sup_{\sigma} Q\left( \Phi_{\lambda, \sigma},\, N\right) ,
\label{c}
\ee
where the supremum is over all single-mode states $\sigma$. Using the explicit form of the Alicki--Fannes--Winter inequality~\cite{Alicki-Fannes, tightuniform} could yield the explicit, rigorous,\footnote{A small note on the meaning of this word for us. The estimate reported here for $c(1/2)$ is found through a numerical search of the zero of a sum of elementary functions, carried out with Wolfram Mathematica. Apart from this numerical step, whose accuracy is guaranteed by the programme's algorithms, it does not involve any other non-analytical insight, such as -- for instance -- `verifying' that a certain function is positive in a certain interval by looking at its graph. A claim of this latter kind would be very far away from a rigorous proof, as it involves keeping under control the infinite number of points that form an interval.} yet very small lower bound $c(1/2) \geq 5.133\times 10^{-6}$. Numerical investigations suggest that this is very far away from the truth, and that one could take at least $c(1/2) \gtrsim 0.066$, which is four orders of magnitude larger than the former estimate. This must be confronted with the `trivial' upper bound descending from Lemma~\ref{universal upper bound Q}, which reads $c(1/2)\leq g(1/2) \approx 1.377$.
\end{rem}

\begin{rem} \label{value c0 rem}
It is perhaps more interesting to look at the slightly different quantities
\bb
c_0(N) \coloneqq \lim_{\lambda\to 0^+} \sup_{\substack{\sigma,\\ 0 < \lambda' \leq \lambda}} Q\left( \Phi_{\lambda',\, \sigma},\, N\right) ,
\label{c0}
\ee
which represent the best-case-scenario quantum communication rates when the transmissivity approaches $0$ but the single-mode environment state $\sigma$ is chosen optimally. Since
\bbb
\lim_{n\to\infty} \frac{(n+1)(n-2)}{4n(n-1)} \left(1-\frac{1}{n} \right)^{n} = \frac{1}{4e}\, , 
\eee
it can be seen that our argument yields the rigorous estimate
\bbb
c_0(1/2) \geq \frac{1}{8 e^2\ln 2} \approx 0.0244\, .
\eee
Numerical investigations produce a substantially higher bound $c_0(1/2) \gtrsim 0.133$, which again must be confronted with the upper bound $c_0(1/2)\leq g(1/2) \approx 1.377$.
\end{rem}

\subsection{Further considerations} \label{further considerations subsec}

It turns out that one can get rid of the multiple options in Proposition~\ref{p sorted prop} if one is willing to exclude the special cases $n=2$ and $n=3$. When this is done something more happens. Namely, the majorisation $p(n,\lambda)\prec q(n,\lambda)$ of Proposition~\ref{majorisation prop} is of a very special type. It actually holds that $p^\uparrow_\ell(n,\lambda) \geq q^\uparrow_\ell(n,\lambda)$ for all $n\geq 4$ and $\frac{1}{n+1}\leq \lambda\leq \frac1n$. Throughout this section we prove these claims.

\begin{lemma} \label{p sorted lemma3}
For all $n\geq 4$,
\bb
p_{n+1}(n,\lambda) \leq p_{n-2}(n,\lambda) \qquad \forall\quad \frac{1}{n+1}\leq \lambda\leq 1\, .
\ee
\end{lemma}

\begin{proof}
Employing the expressions~\eqref{p ell}, we see that
\bbb
\frac{p_{n-2}(n,\lambda)}{p_{n+1}(n,\lambda)} = \frac{n(n-1)}{6(n+1)} \frac{\lambda}{(1-\lambda)^2} \left( 3 + \frac{1}{1-\lambda}\left( 3 - (n+1)\lambda\right)^2 \right) \eqqcolon \frac{n(n-1)}{6(n+1)}\, g_{n}(\lambda)\, .
\eee
It is not difficult to see that
\bbb
\frac{dg_n(\lambda)}{d\lambda} = \frac{6}{(1-\lambda)^4}\left(\frac12 (n^2-2)\lambda^2 - (2n-1)\lambda + 2 \right) \geq 0 \qquad \forall\quad 0\leq \lambda\leq 1\, ,
\eee
because the discriminant of the second-degree polynomial on the right-hand side equals $9-4n$ and is therefore negative as long as $n\geq 3$. Thus,
\bbb
\min_{\frac{1}{n+1}\leq\lambda\leq 1} \left\{\frac{p_{n-2}(n,\lambda)}{p_{n+1}(n,\lambda)} -1\right\} = \frac{n(n-1)}{6(n+1)}\, g_n\left( \frac{1}{n+1}\right) -1 =\frac{1}{6n^2}(n+1)(n-4) \geq 0
\eee
for all $n\geq 4$.
\end{proof}

\begin{prop} \label{p sorted prop n>3}
For all $n\geq 4$, 
\bb
p^\uparrow(n,\lambda) = \Big( p_0(n,\lambda),\, p_1(n,\lambda), \ldots,\, p_{n-3}(n,\lambda),\, p_n(n,\lambda),\, p_{n+1}(n,\lambda),\, p_{n-2}(n,\lambda),\, p_{n-1}(n,\lambda) \Big) \qquad \forall \ \frac{1}{n\!+\!1}\leq \lambda\leq \widetilde{\lambda}_+(n)\, ,
\label{p sorted n>3}
\ee
where $\widetilde{\lambda}_+(n)$ is defined by~\eqref{lambda+ tilde}. In other words, for the stated range of values of $\lambda$ the probability vector $p(n,\lambda)$ can be sorted in ascending order by exchanging the last two pairs of entries.
\end{prop}

\begin{proof}
It suffices to combine Lemmata~\ref{p sorted lemma1},~\ref{p sorted lemma2} and~\ref{p sorted lemma3}. Note that $\widetilde{\lambda}_+(n)\leq \lambda_+(n)$ for all $n\geq 4$. This can be shown e.g.\ by noting that
\bbb
\widetilde{\lambda}_+(n) \leq \frac{3-\sqrt3}{n+2} \leq \lambda_+(n)\qquad \forall\ n\geq 18\, ,
\eee
where the first relation is equivalent to a linear inequality upon elementary algebraic manipulations, while the second is easily seen to hold for all $n\geq 1$ by direct inspection of~\eqref{lambda+}. In the remaining cases $n=4,\ldots, 17$, the fact that $\widetilde{\lambda}_+(n)\leq \lambda_+(n)$ can be checked numerically.
\end{proof}

Now that the probability distribution $p(n,\lambda)$ has been sorted in ascending order by a fixed permutation, we proceed to check that indeed $p^\uparrow_\ell(n,\lambda) \geq q^\uparrow_\ell(n,\lambda)$ for all $n\geq 4$ and $\frac{1}{n+1}\leq \lambda\leq \frac1n$.

\begin{lemma} \label{p vs q sorted n-2 lemma}
For all $n\geq 2$,
\bb
q_{n-2}(n,\lambda) \leq p_n(n,\lambda) \qquad \forall\quad \frac{1}{n+1}\leq \lambda\leq \frac1n\, .
\label{p vs q sorted n-2}
\ee
\end{lemma}

\begin{proof}
One verifies that
\begin{align*}
&\frac{\partial}{\partial \lambda} \left( p_n(n,\lambda) - q_{n-2}(n,\lambda) \right) \\
&\qquad = \frac{1}{12} (1-\lambda )^{n-3} \left(2 - (n+1)\lambda \right) \left(6-3 (n+6) \lambda - 3 \left(n^2 - 4n - 6\right) \lambda^2 + (n+1)(n+2)(n-3) \lambda^3\right) \\
&\qquad \eqqcolon \frac{1}{12} (1-\lambda)^{n-3} \left(2 - (n+1)\lambda \right) h_n(\lambda)\, .
\end{align*}
We will now show that $h_n(\lambda)\geq 0$ for all $\frac{1}{n+1}\leq \lambda\leq \frac1n$. To this end, compute
\bbb
\frac16 \frac{d^2 h_n(\lambda)}{d\lambda} = - \left(n^2 - 4n - 6\right) + (n+1)(n+2)(n-3) \lambda \geq 0 \qquad \forall\quad \frac{1}{n+1}\leq \lambda\leq \frac1n\, ,
\eee
where the last inequality holds because: (i)~it can be verified explicitly for $n=2$ and $n=3$; (ii)~for $n\geq 4$, one has that
\bbb
\frac{n^2 - 4n - 6}{(n+1)(n+2)(n-3)}\leq \frac{1}{n+1} \qquad \forall\ n\geq 4\, ,
\eee
with equality for $n=4$. Since we have shown that $\frac{d h_n(\lambda)}{d\lambda}$ is increasing in $\lambda$ on $\left[\frac{1}{n+1},\, \frac1n\right]$, there it holds that
\begin{align*}
\frac13 \frac{d h_n(\lambda)}{d\lambda} &\leq \frac13 \frac{d h_n(\lambda)}{d\lambda} \Big|_{\lambda = 1/n} \\
&= \left(-(n+6) - 2\left(n^2 - 4n -6\right) \lambda +(n+1)(n+2)(n-3)\lambda^2\right) \Big|_{\lambda = 1/n} \\
&= -2(n-1) + \frac{5}{n} -\frac{6}{n^2} \leq - (n-1)\left( 2 -\frac{5}{n^2}\right) \leq 0 \, .
\end{align*}
Thus, $h_n(\lambda)$ is decreasing in $\lambda$ on $\left[\frac{1}{n+1},\, \frac1n\right]$. From this we deduce that
\bbb
h_n(\lambda) \geq h_n\left( \frac{1}{n}\right) = \frac{(n-1)(n-2)(n-3)}{n^3} \geq 0
\eee
for all $n=2,3,4,\ldots$.
\end{proof}

\begin{lemma} \label{p vs q sorted n-1 lemma}
For all $n\geq 2$,
\bb
q_{n-1}(n,\lambda) \leq p_{n+1}(n,\lambda) \qquad \forall\quad \frac{1}{n+1}\leq \lambda\leq \frac{2}{n+2}\, .
\label{p vs q sorted n-1}
\ee
\end{lemma}

\begin{proof}
A simple calculation shows that
\bbb
2(n+1)(1-\lambda)^2\left( 1-\frac{q_{n-1}(n,\lambda)}{p_{n+1}(n,\lambda)}\right) = (n-1)\left((n+1)\lambda-1\right) \left( 2 -(n+2)\lambda\right) \geq 0\qquad \frac{1}{n+1}\leq \lambda\leq \frac{2}{n+2}\, ,
\eee
completing the proof.
\end{proof}

\begin{lemma} \label{p vs q sorted n lemma}
For all $n\geq 4$,
\bb
q_n(n,\lambda)\leq p_{n-2}(n,\lambda)\qquad \forall\quad \frac{1}{n+1}\leq \lambda\leq \frac1n\, .
\label{p vs q sorted n}
\ee
\end{lemma}

\begin{proof}
One finds that
\begin{align*}
&\frac{n(n-1)}{6}\, \lambda^2 \left(3(1-\lambda )+(3-(n+1)\lambda )^2\right) \left( 1 - \frac{q_n(n,\lambda)}{p_{n-2}(n,\lambda)} \right) \\
&\qquad = -1+ (n+4) \lambda + \left(n^2 - 6n -6\right) \lambda^2 - (n+1) \left( n^2 - \frac52\, n - 4 \right) \lambda^3 + \frac16\, (n+1)^2 (n+2)(n-3) \lambda^4 \\
&\qquad \eqqcolon r_n(\lambda)\, .
\end{align*}
We look at the polynomial $r_n(\lambda)$ and its derivatives in the interval $\left[ \frac{1}{n+1}, \frac1n\right]$. Since $\frac{d^4 r_n(\lambda)}{d\lambda} = 4(n+1)^2(n+2)(n-3)\geq 0$, the function $\frac{d^2 r_n(\lambda)}{d\lambda}$ is convex. Therefore, on the larger interval $\left[ \frac{1}{n+1},\, \frac{2}{n+1}\right]\supset \left[ \frac{1}{n+1}, \frac1n\right]$ it holds that
\begin{align*}
\frac{d^2 r_n(\lambda)}{d\lambda} &\leq \max\left\{ \frac{d^2r_n(\lambda)}{d\lambda} \Big|_{\lambda=1/(n+1)},\, \frac{d^2 r_n(\lambda)}{d\lambda} \Big|_{\lambda=2/(n+1)}\right\} \\
&= \max \left\{ - n(2n-1),\, - 2(n-2)(n-3) \right\} \\
&\leq 0\, .
\end{align*}
In turn, this tells us that $r_n(\lambda)$ is concave. Thus, on $\left[ \frac{1}{n+1}, \frac1n\right]$ it holds that
\begin{align*}
r_n(\lambda) &\leq \max\left\{ r_n\left(\frac{1}{n+1}\right),\, r_n\left(\frac1n\right) \right\} \\
&= \min\left\{ \frac{n(n-4)}{6(n+1)^2},\, \frac{(n-1)(n-2)(n^2+n-3)}{6n^4} \right\} \\
&\geq 0\, .
\end{align*}
This proves the claim.
\end{proof}

We are finally ready to prove our last claim.

\begin{prop} \label{majorisation n>3 prop}
Let $n\geq 4$ be an integer. Then
\bb
p^\uparrow_\ell(n,\lambda) \geq q^\uparrow_\ell(n,\lambda)\qquad \forall\quad \ell=0,\ldots, n\, ,\qquad \forall\quad \frac{1}{n+1}\leq \lambda\leq \frac1n\, ,
\label{majorisation n>3}
\ee
with the reverse inequality holding instead for $\ell=n+1$. In particular, $p(n,\lambda)\prec q(n,\lambda)$ for all $\frac{1}{n+1}\leq \lambda\leq \frac1n$.
\end{prop}

\begin{proof}
Since $\widetilde{\lambda}_+(n)\geq \frac1n$ for all $n\geq 1$, Proposition~\ref{p sorted prop n>3} applies and tell us that the ordering of $p(n,\lambda)$ is as in~\eqref{p sorted n>3}. Now, the cases $\ell=0,\ldots,n-3$ of~\eqref{majorisation n>3} follow from Lemma~\ref{p ell < q ell}, as usual. When $\ell=n-2$, we have instead to verify that $p_n(n,\lambda) \geq q_{n-2}(n,\lambda)$, which is a consequence of Lemma~\ref{p vs q sorted n-2 lemma}. For $\ell=n-1$, the claim amounts to the inequality $p_{n+1}(n,\lambda) \geq q_{n-1}(n,\lambda)$, which holds by Lemma~\ref{p vs q sorted n-1 lemma}, because $\frac{2}{n+2}\geq \frac1n$ whenever $n\geq 2$. The last case is $\ell=n$, for which we have to show that $p_{n-2}(n,\lambda)\geq q_{n}(n,\lambda)$; this is guaranteed to hold by Lemma~\ref{p vs q sorted n lemma}. The reverse inequality holds for $\ell=n+1$ by normalisation:
\bbb
p^\uparrow_{n+1}(n,\lambda) = 1 - \sum_{\ell=0}^n p^\uparrow(n,\lambda) \geq 1 - \sum_{\ell=0}^n q^\uparrow (n,\lambda) = q^\uparrow_{n+1}(n,\lambda)\, .
\eee
Finally, majorisation follows by direct inspection.
\end{proof}

\section{Some extensions of Theorem~\ref{Behemoth thm}} \label{extensions subsec}

Throughout this section, we discuss some possible extensions of Theorem~\ref{Behemoth thm}. In particular, we look into the case where the transmission channel results from a concatenation of multiple beam splitters instead of a single one. The scenario we consider is depicted in Figure~\ref{entangled composition fig}.

\begin{figure}[ht]
\centering
\includegraphics[scale=0.8]{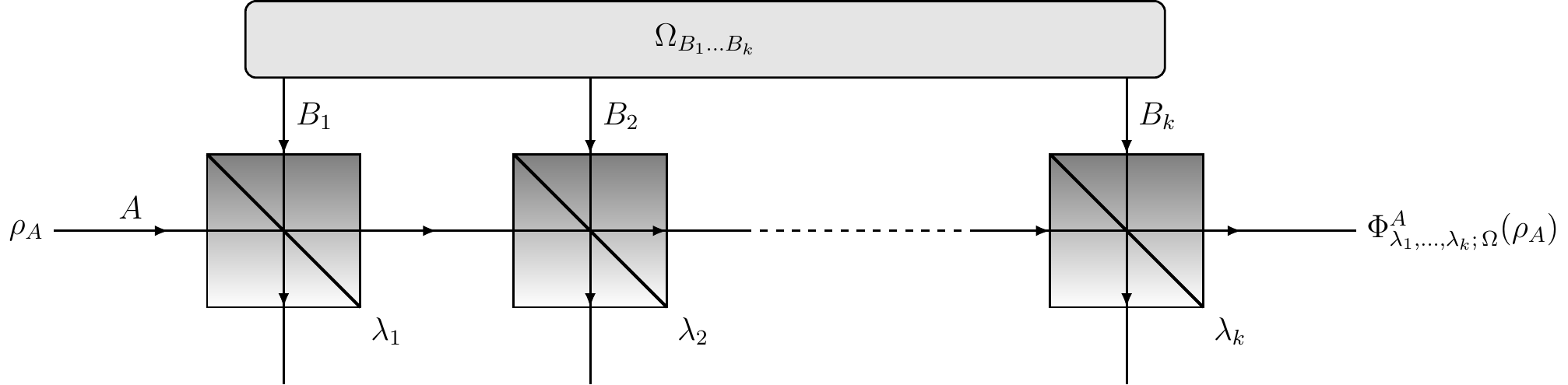}
\caption{The input state $\rho_A$ is sent through a sequence of beam splitters with transmissivities $\lambda_1, \ldots, \lambda_k$. The vertical output arms coming down from each beam splitter are simply traced away. The action of the channel $\Phi_{\lambda_1, \ldots, \lambda_k;\, \Omega}$ is given by~\eqref{entangled channel}.}
\label{entangled composition fig}
\end{figure}

As usual, we will fix the values of the transmissivities and optimise the environment state so as to maximise the quantum capacity of the resulting channel. There are at least three different scenarios we could look into:
\begin{enumerate}[(a)]
\item $\Omega_{B_1\ldots B_k}$ is a generic multipartite entangled state;
\item $\Omega_{B_1\ldots B_k}$ is constrained to be either bi-separable or fully separable;
\item $\Omega_{B_1\ldots B_k}$ is constrained to be a product state.
\end{enumerate}
We will now show that in setting~(a) our Theorem~\ref{Behemoth thm} still holds. We need a preliminary lemma.

\begin{lemma} \label{entangled composition lemma}
Let $0\leq \lambda_1, \ldots, \lambda_k\leq 1$ be transmissivities, and let $\sigma$ be an arbitrary $m$-mode state. Set $\lambda\coloneqq \lambda_1\ldots \lambda_k$. Then there exists a $k$-partite state $\Omega_{B_1\ldots B_k}$ of $k$ systems of $m$ modes each such that the channel $\Phi_{\lambda_1, \ldots, \lambda_k;\, \Omega}$ defined by
\bb
\Phi_{\lambda_1, \ldots, \lambda_k;\, \Omega}^A(\rho_A) \coloneqq \Tr_{B_1\ldots B_k} \left[ U_{\lambda_k}^{AB_k}\ldots\, U_{\lambda_1}^{AB_1} \left( \rho_A \otimes \Omega_{B_1\ldots B_k} \right) \left( U_{\lambda_1}^{AB_1} \right)^\dag \ldots \left( U_{\lambda_k}^{AB_k} \right)^\dag \right]
\label{entangled channel}
\ee
satisfies that
\bb
\Phi_{\lambda, \sigma} = \Phi_{\lambda_1, \ldots, \lambda_k;\, \Omega}\, .
\label{entangled composition}
\ee
\end{lemma}

\begin{proof}
Iterating the same argument as in the proof of Lemma~\ref{data processing attenuators lemma} shows that
\bbb
\chi_{\Phi_{\lambda_1, \ldots, \lambda_k;\, \Omega}(\rho)} (\alpha) = \chi_{\rho}\left( \sqrt{\lambda_1 \lambda_2 \ldots \lambda_k}\, \alpha\right) \chi_{\Omega} \left( \sqrt{(1-\lambda_1)\lambda_2\ldots \lambda_k}\, \alpha,\, \sqrt{(1-\lambda_2)\lambda_3\ldots \lambda_k}\, \alpha,\, \ldots,\, \sqrt{1-\lambda_k}\, \alpha\right) ,
\eee
where $\alpha\in \C^m$ is generic. This coincides with $\chi_{\Phi_{\lambda, \sigma}(\rho)}(\alpha) = \chi_\rho\left( \sqrt\lambda\, \alpha\right) \chi_\sigma\left( \sqrt{1-\lambda}\, \alpha\right)$ for all $\rho$ if
\bb
\chi_{\Omega} \left( \sqrt{\frac{(1-\lambda_1) \lambda_2\ldots \lambda_k}{1-\lambda_1\ldots \lambda_k}}\, \alpha,\, \sqrt{\frac{(1-\lambda_2) \lambda_3\ldots \lambda_k}{1-\lambda_1\ldots \lambda_k}}\, \alpha, \ldots, \sqrt{\frac{1- \lambda_k}{1-\lambda_1\ldots \lambda_k}}\, \alpha \right) = \chi_\sigma\left( \alpha \right) .
\label{constraint Omega}
\ee
To construct a state $\Omega_{B_1\ldots B_k}$ such that~\eqref{constraint Omega} holds, we start by defining the numbers
\bb
\eta_i \coloneqq \frac{\lambda_i (1-\lambda_1\ldots \lambda_{i-1})}{1-\lambda_1\ldots \lambda_i}\, ,\qquad \forall\ i=2,\ldots, k\, .
\label{eta i}
\ee
We can then rephrase~\eqref{constraint Omega} as
\bb
\chi_{\Omega} \left( \sqrt{\eta_2\ldots \eta_k}\, \alpha, \, \sqrt{(1-\eta_2)\eta_3\ldots \eta_k}\, \alpha,\, \sqrt{(1-\eta_3)\eta_4\ldots \eta_k}\, \alpha, \ldots, \, \sqrt{1-\eta_k}\, \alpha\right) = \chi_\sigma(\alpha)\, .
\label{constraint Omega 2}
\ee
Now we observe that~\eqref{constraint Omega 2} is satisfied if and only if
\bb
\sigma_{B_1} = \Tr_{B_2\ldots B_k} \left[ U_{\eta_k}^{B_1B_k} \ldots\, U_{\eta_2}^{B_1B_2}\, \Omega_{B_1\ldots B_k} \left( U_{\eta_2}^{B_1B_2} \right)^\dag \ldots \left( U_{\eta_k}^{B_1B_k}\right)^\dag \right] .
\label{constraint Omega 3}
\ee
In order to meet~\eqref{constraint Omega 3}, it suffices to set e.g.
\bb
\Omega_{B_1\ldots B_k} \coloneqq \left( U_{\eta_2}^{B_1B_2} \right)^\dag \ldots \left( U_{\eta_k}^{B_1B_k}\right)^\dag \left( \sigma_{B_1} \otimes \ketbra{0}_{B_2\ldots B_k} \right) U_{\eta_k}^{B_1B_k} \ldots U_{\eta_2}^{B_1B_2}\, .
\label{Omega}
\ee
This construction concludes the proof. Incidentally, note that we could replace $\sigma_{B_1} \otimes \ketbra{0}_{B_2\ldots B_k}$ in~\eqref{Omega} with any extension of $\sigma_{B_1}$.
\end{proof}

\begin{rem}
The argument in the above proof amounts to an equivalence between the channel depicted in Figure~\ref{entangled composition fig} and that in Figure~\ref{equivalent network fig} below.
\end{rem}

\begin{figure}[ht]
\centering
\includegraphics[scale=0.8]{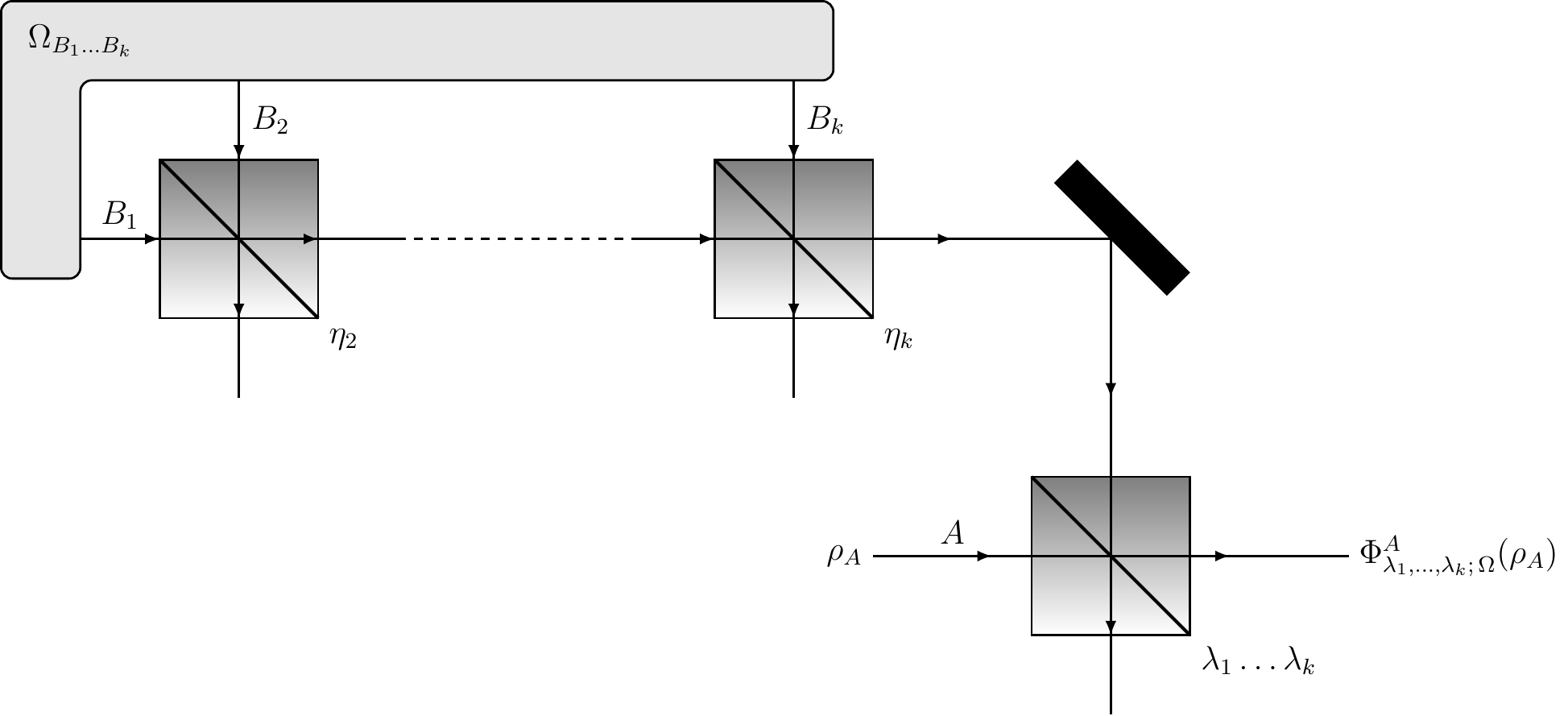}
\caption{An equivalent representation of the channel in Figure~\ref{entangled composition fig}, as constructed in the above proof of Lemma~\ref{entangled composition lemma}. The transmissivities $\eta_i$ ($i=2,\ldots, k$) are given by~\eqref{eta i}.}
\label{equivalent network fig}
\end{figure}

We are now ready to formulate the sought extension of Theorem~\ref{Behemoth thm}.

\begin{thm} \label{extended Behemoth thm}
Let $0<\lambda_1,\ldots, \lambda_k \leq 1$ be positive transmissivities. There exists a state $\Omega_{B_1\ldots B_k}$ of $k$ single-mode systems such that the channel $\Phi_{\lambda_1, \ldots, \lambda_k;\, \Omega}$ defined by~\eqref{entangled channel} (see Figure~\ref{entangled composition fig}) satisfies that
\bb
Q\big(\Phi_{\lambda_1, \ldots, \lambda_k;\, \Omega}\big) \geq Q\big(\Phi_{\lambda_1, \ldots, \lambda_k;\, \Omega},\, 1/2 \big) \geq c\, ,
\ee
where $c>0$ is the same universal constant as in Theorem~\ref{Behemoth thm}.
\end{thm}

\begin{proof}
Thanks to Lemma~\ref{entangled composition lemma}, by varying the state $\Omega$ we can reproduce all general attenuator channels $\Phi_{\lambda_1\ldots \lambda_k,\, \sigma}$. Noting that $\lambda_1\ldots \lambda_k>0$, it suffices to apply Theorem~\ref{Behemoth thm} to conclude.
\end{proof}

As a final remark, let us point out that, while Theorem~\ref{extended Behemoth thm} clarifies the situation in scenario~(a) above, i.e.\ when $\Omega_{B_1\ldots B_k}$ is a generic multipartite entangled state, it would also be of interest to explore scenarios~(b) and~(c). We leave this for future investigations.

%\bibliographystyle{apsrev4-1}
%\bibliography{../../biblio}

\end{document}

%% file: newdef.tex
\usepackage[latin1]{inputenc}
\usepackage{amsthm}
\usepackage{amssymb}
\usepackage{amsmath}
\usepackage{bbold}
\usepackage{bbm}
\usepackage{nonfloat}
\usepackage[pdftex]{hyperref}
\usepackage{braket}
\usepackage{dsfont}
\usepackage{mathdots}
\usepackage{mathtools}
\usepackage{enumerate}
\usepackage{csquotes}
\usepackage{stmaryrd}
\usepackage[cal=boondox]{mathalfa}
\usepackage{graphicx}
\usepackage{stackengine}
\usepackage{scalerel}
\usepackage{array}
\usepackage{makecell}
\newcolumntype{x}[1]{>{\centering\arraybackslash}p{#1}}
\usepackage{tikz}
\usepackage{pgfplots}
\usetikzlibrary{shapes.geometric, shapes.misc, positioning, arrows, decorations.pathreplacing, angles, quotes}
\usepackage{booktabs}
\usepackage{xfrac}
\usepackage{siunitx}
\usepackage{centernot}
\usepackage{comment}

\newtheorem{thm}{Theorem}
\newtheorem*{thm*}{Theorem}
\newtheorem{prop}[thm]{Proposition}
\newtheorem*{prop*}{Proposition}
\newtheorem{lemma}[thm]{Lemma}
\newtheorem*{lemma*}{Lemma}
\newtheorem{cor}[thm]{Corollary}
\newtheorem*{cor*}{Corollary}

\newtheorem*{cj*}{Conjecture}

\newtheorem*{Def*}{Definition}

\makeatletter
\def\thmhead@plain#1#2#3{%
  \thmname{#1}\thmnumber{\@ifnotempty{#1}{ }\@upn{#2}}%
  \thmnote{ {\the\thm@notefont#3}}}
\let\thmhead\thmhead@plain
\makeatother

\theoremstyle{definition}
\newtheorem{rem}[thm]{Remark}

\newtheorem{ex}[thm]{Example}

\newtheorem{manualthminner}{Theorem}
\newenvironment{manualthm}[1]{%
  \renewcommand\themanualthminner{#1}%
  \manualthminner
}{\endmanualthminner}

\newcommand{\bb}{\begin{equation}}
\newcommand{\bbb}{\begin{equation*}}
\newcommand{\ee}{\end{equation}}
\newcommand{\eee}{\end{equation*}}
\newcommand*{\coloneqq}{\mathrel{\vcenter{\baselineskip0.5ex \lineskiplimit0pt \hbox{\scriptsize.}\hbox{\scriptsize.}}} =}
\newcommand*{\eqqcolon}{= \mathrel{\vcenter{\baselineskip0.5ex \lineskiplimit0pt \hbox{\scriptsize.}\hbox{\scriptsize.}}}}

\newcommand{\texteq}[1]{\stackrel{\mathclap{\scriptsize \mbox{#1}}}{=}}

\newcommand{\textgeq}[1]{\stackrel{\mathclap{\scriptsize \mbox{#1}}}{\geq}}
\newcommand{\ketbra}[1]{\ket{#1}\!\!\bra{#1}}
\newcommand{\ketbraa}[2]{\ket{#1}\!\!\bra{#2}}
\newcommand{\sumno}{\sum\nolimits}

\newcommand{\G}{\mathrm{\scriptscriptstyle G}}

\newcommand{\id}{\mathds{1}}
\newcommand{\R}{\mathds{R}}
\newcommand{\N}{\mathds{N}}
\newcommand{\C}{\mathds{C}}

\DeclareMathOperator{\Tr}{Tr}

\DeclareMathOperator{\Span}{span}
\DeclareMathAlphabet{\pazocal}{OMS}{zplm}{m}{n}

\DeclareMathOperator{\tr}{tr}

\newcommand{\HH}{\pazocal{H}}

\newcommand{\lsmatrix}{\left(\begin{smallmatrix}}
\newcommand{\rsmatrix}{\end{smallmatrix}\right)}

\stackMath

\stackMath

\makeatletter
\newcommand*\rel@kern[1]{\kern#1\dimexpr\macc@kerna}
\newcommand*\widebar[1]{%
  \begingroup
  \def\mathaccent##1##2{%
    \rel@kern{0.8}%
    \overline{\rel@kern{-0.8}\macc@nucleus\rel@kern{0.2}}%
    \rel@kern{-0.2}%
  }%
  \macc@depth\@ne
  \let\math@bgroup\@empty \let\math@egroup\macc@set@skewchar
  \mathsurround\z@ \frozen@everymath{\mathgroup\macc@group\relax}%
  \macc@set@skewchar\relax
  \let\mathaccentV\macc@nested@a
  \macc@nested@a\relax111{#1}%
  \endgroup
}